\documentclass[journal]{IEEEtran}
\usepackage{graphicx}
\usepackage{cite}
\usepackage{xcolor}
\usepackage{amsfonts}
\usepackage{amsmath}
\usepackage{amssymb}
\usepackage{mathrsfs}
\usepackage{bbm}
\usepackage{fix-cm}
\usepackage{textcomp}
\usepackage[center]{subfigure}
\usepackage{euscript}
\usepackage{setspace}
\usepackage{subfigure}
\usepackage[justification=centering]{caption}
\usepackage{float}
\usepackage{stfloats}
\usepackage{lipsum}%生成随机文本

\newtheorem{corollary}{Corollary}
\newtheorem{assumption}{Assumption}
\newtheorem{remark}{Remark}
\newtheorem{definition}{Definition}

\newtheorem{lemma}{Lemma}
\newtheorem{theorem}{Theorem}

\newtheorem{example}{Example}
\newtheorem{fac}{Fact}
\ifCLASSINFOpdf

\else

\begin{document}
\title{%New results on Hessian Matrices and Stabilization for It\^o Stochastic T-S Models
%Line Integral Approach to   It\^o  Stochastic  T-S Models}
Stabilization of It\^o Stochastic T-S Models via Line Integral and Novel Estimate for  Hessian Matrices}

\author{Shaosheng Zhou,  Yingying Han, Baoyong Zhang
\thanks{
This work was supported  by
the National Natural Science Foundation of P. R. China
under Grants 61673149.

Corresponding author. Shaosheng Zhou is with
 Institute
of Information and Control, Department of Automation, Hangzhou
Dianzi University, Hangzhou Zhejiang 310018, P. R. China (e-mail:
sszhou65@163.com).

Yingying Han is with Institute
of Information and Control, Department of Automation, Hangzhou
Dianzi University, Hangzhou Zhejiang 310018, P. R. China (e-mail:
yyhan0826@163.com).

Baoyong Zhang is with the School of Automation,
Nanjing University of Science and  Technology, Nanjing Jiangsu  210094 , P. R.
China (e-mail: baoyongzhang@njust.edu.cn).
}}
\maketitle

\begin{abstract}
This paper proposes a line integral Lyapunov function approach to
stability analysis and stabilization for  It\^o stochastic T-S models. Unlike the deterministic case, stability analysis of
 this model  needs  the information of Hessian
matrix of the line integral Lyapunov  function which is related to  partial derivatives of the  basis
 functions. By  introducing a new method to handle these partial derivatives and using the property of state-dependent matrix with rank one, the stability
conditions of the underlying system  can be established via a line integral Lyapunov function. These conditions obtained are more general  than  the ones which are based on quadratic Lyapunov functions.
Based on the stability conditions,  a controller is developed
 by cone complementarity linerization algorithm.
A  non-quadratic Lyapunov function approach is thus proposed for the stabilization problem of the  It\^o
stochastic T-S models. It has been shown that the problem can be solved by optimizing  sum of traces for a group of products of matrix variables with linear
constraints. Numerical examples are given to illustrate the
effectiveness of the proposed control scheme.
  \end{abstract}

 \ifCLASSOPTIONpeerreview
 \begin{center} \bfseries EDICS Category: 3-BBND \end{center}
 \fi
 %\footnotetext[ ] { This work was supported  by
%the National Natural Science Foundation of P. R. China
%under Grants 61673149.}
%\footnotetext[ ] {
%Corresponding author. Shaosheng Zhou is with
% Institute
%of Information and Control, Department of Automation, Hangzhou
%Dianzi University, Hangzhou Zhejiang 310018, P. R. China (e-mail:
%sszhou65@163.com).}
%
%\footnotetext[ ] {  Yingying Han is with Institute
%of Information and Control, Department of Automation, Hangzhou
%Dianzi University, Hangzhou Zhejiang 310018, P. R. China (e-mail:
%yyhan0826@163.com).}
\IEEEpeerreviewmaketitle

\section{Introduction}
The  past few decades have witnessed the significant advances of Takagi-Sugeno (T-S) fuzzy systems, see, for example \cite{feng06,PaoTsunLin2009,sa07,ta85} and references therein. The
main motivation should be attributed to their capabilities that allow them to represent and handle some sort of qualitative information  such as  operators'  experience, experts' knowledge and so forth involved in the plants. The practical experience of these skilled operators and experts' knowledge are usually summarized and expressed in language, and fuzzy set theory can  convert  such sort of qualitative information  into quantitative data in light of membership functions. For stability analysis and stabilization  of  fuzzy systems, many researchers have resorted to a  quadratic Lyapunov function approach \cite{CHFang2006,k2000,zhouren2010}. The main drawback concerning this approach lies in the fact that it is difficult for a common positive definite matrix to satisfy the stability conditions of all local linear models. In this situation, the non-quadratic Lyapunov function is a helpful alternative \cite{gv04,gb2012,LAMozelli2009,SZhou2005,zhou07}.  It has been  shown that the stability and stabilization results based on non-quadratic Lyapunov function are less conservative than those based on common quadratic ones. It is worth noting that unlike their discrete-time counterpart, continuous-time T-S models induce the problem of handling derivatives of the fuzzy basis functions that emerge when the basis dependent Lyapunov function is employed to derive stability and stabilization results \cite{gb2012,KTanaka2003,BZhang2007}. In such case, the upper bounds of the time derivatives of the fuzzy basis functions are used to analyze stability performance and to develop a controller. Since the time derivatives of fuzzy basis functions are related to the vector field that governs the T-S
models and are dependent on the state and control of the systems, their upper bounds are not always readily available, and sometimes they  do not even exist. To overcome these difficulties, the authors in \cite{RW2006} proposed a novel fuzzy Lyapunov function that was formulated as a line integral of a vector along a path from the origin to the current state. By using this line  integral function, the stability results did not
involve the upper bound of the time derivatives of fuzzy basis functions. The main issues in \cite{RW2006} are that the T-S fuzzy model have a specific premise structure and  that the stabilization results are provided by bilinear  matrix inequalities.

When a fuzzy system is disturbed by random perturbations, the system turns to be a stochastic fuzzy one. Recently, T-S fuzzy descriptions have been introduced to study the analysis and synthesis problem of  stochastic nonlinear systems \cite{gfwq2013,LHu2009,PuyinLiua2005,BZhang2007}. The local uniqueness results for the solutions of fuzzy stochastic differential equations were provided in \cite{TMalinowski2015}. For uncertain nonlinear stochastic time-delay systems, a fuzzy controller was developed to guarantee the robust asymptotic stability and attenuation performance \cite{chang2010}. In \cite{sheng2009}, by using Lyapunov method and stochastic analysis approaches, the
delay-dependent stability criterion was presented for T-S fuzzy
Hopfield neural networks with parametric uncertainties and stochastic
perturbations. For a class of T-S model based stochastic
systems, the stabilization problem was investigated in
\cite{zhouren2010}.  It should be pointed out that most
analysis or synthesis results of the T-S model based stochastic
systems mentioned above are obtained  by using or partly using a
common quadratic Lyapunov function.

Motivated by \cite{RW2006} and \cite{zhouren2010}, this paper
proposes a line integral Lyapunov function  approach to
stability analysis and stabilization for a class of It\^o stochastic  T-S model based nonlinear systems.
 By using the line integral  Lyapunov function
proposed in \cite{RW2006}, more general sufficient conditions for the stochastic
asymptotic stability of the  class of systems can be established. The control design
is facilitated by the introduction of some additional matrix
variables and a cone complementarity linearization algorithm. These variables
decouple the coupling terms in the obtained stability conditions and
make the control design feasible. It will show that the solution of
the controller design problem can be obtained by solving a
minimization problem with linear constraints.

The organization of this paper is as follows. Section 2 formulates
the problem and presents some preliminary results. In Section 3,
stochastic stability analysis is given. Controller design is presented  in
Section 4 and a numerical example
is given in Section 5 to illustrate the effectiveness of the proposed
approach. Finally, the paper is concluded in Section 6.

%\emph{Notation.}
Throughout this paper, a real symmetric matrix
$P>0\  (\geq 0)$ denotes $P$ being a positive definite (or positive
semidefinite) matrix, and $A>(\geq )B$ means $A-B>(\geq )0$. $I$ is
used to denote an identity matrix with proper dimensions. The trace
of a square matrix $A$ is denoted by tr$A$  and the notation $A^{\tau}$ represents
the transpose of the matrix $A$. The expression $A^{S}$ stands
for the sum of $A$ and $A^{\tau}$. For symmetric matrices, we use $(*)$
as an ellipsis for a block matrix that is induced by symmetry. The
notation $(\Omega,\mathcal {F},\mathcal{P})$ represents the
probability space with $\Omega$ the sample space, $\mathcal {F}$ the
$\sigma$-algebra of subsets of the sample space and $\mathcal{P}$
the probability measure on $\mathcal {F}$. Matrices, if not
explicitly stated, are assumed to be compatible for algebraic
operations.

\section{Problem Formulation and Preliminaries}
Consider an It\^o stochastic  T-S model
described by the following rules:\\
\textbf{Plant Rule} $ i$:
\textbf{If}  $x_{1}$ is $ F_{1}^{\alpha_{i1}}$,
$\cdots $ , and $x _{n}$ is $F_{n}^{\alpha_{in}}$, \textbf{then}
{\setlength\abovedisplayskip{6pt}
\setlength\belowdisplayskip{5pt}
 \[
 dx=\left(A_{i}x + B_{i}u \right)dt+C_{i}x dW(t)
 \]
 where $ i=1,\cdots, s$, $s$ is the total number
of fuzzy rules; $x=\left(x_{1},\ldots,x_{n}\right)^{\tau} \in
\mathbb{R}^{n}$, $u \in \mathbb{R}^{p}$ are the state vector and the
control input, respectively;
 for each pair $ (i,j )\in \{1,\cdots ,s\} \times \{1,\cdots ,n\}$, $F_{j}^{\alpha_{ij}}$ is the
fuzzy set in the  $i$th rule based on the premise variable $x _{j}$
with $\alpha_{ij}$ specifying which $x _{j}$-based fuzzy set is used
in the $i$th If-Then rule.  $W(t)$ is a Wiener process on a complete
probability space $(\Omega,\mathcal {F},\mathcal{P})$ adapted to a
filtration $\{\mathcal{F}_{t}\}_{t\geq0}$; $A_{i} \in \mathbb{R}^{n
\times n}$, $B_{i} \in \mathbb{R}^{n \times p}$ and $C_{i}\in
\mathbb{R}^{n \times n}$ are known real constant matrices.
In what follows, for notational convenience, we denote  the two sets
$\{1,\cdots,s\}$ and $\{1,\cdots,n\}$ by $\mathcal{S}$ and $\mathcal{N}$
respectively. Assume that for any $j\in \mathcal{N}$, the number of
all the $x_j$-based fuzzy  sets used in the  $s$ fuzzy rules is
$s_j$ with  $s_j$  a certain  positive integer. The superscript
$\alpha_{ij}$ in the fuzzy set $F_{j}^{\alpha_{ij}}$ is an ordered
index of all the  $s_j$ $x_j$-based fuzzy  sets
$\left\{F_{j}^{\alpha_{ij}}\right \}_{\alpha_{ij}=1}^{s_j}$ which
will play a role in the following development. If
$\alpha_{ij}=\rho_j,\ i \in \mathcal{S}$, which means that in the
$i$th fuzzy rule the premise variable $x _{j}(t)$ belongs to the
$\rho_j$th fuzzy set $F_{j}^{\rho_j}$, one can see that
\begin{equation*}
 1\leq\rho_j \leq s_j,\  j \in \mathcal{N}.
 \end{equation*}
For $i\in \mathcal{S}$, $j\in
\mathcal{N}$ and $x_j \in \mathbb{R}$, let
$w_{j}^{\alpha_{ij}}(x_j)$ be the membership function of
 fuzzy set $F_{j}^{\alpha_{ij}}$. The normalized membership functions are defined
by
\begin{equation}
\mu_{j}^{\alpha_{ij}}(x_{j})=\frac{w_{j}^{\alpha_{ij}}(x_j)}{\Sigma_{\alpha_{ij}=1}^{s_j}
w_{j}^{\alpha_{ij}}(x_j)},
     \ \ i\in \mathcal{S},\ \ j\in \mathcal{N}. \label{normalization}
\end{equation}
which satisfy
{\setlength\abovedisplayskip{6pt}
\setlength\belowdisplayskip{5pt}
\begin{equation}
0\leq \mu_{j}^{\alpha_{ij}}(x_j)\leq 1,\ \
\sum_{\alpha_{ij}=1}^{s_j}
\mu_{j}^{\alpha_{ij}}(x_j)=1.\label{small0}
\end{equation}}
Then, the fuzzy basis functions can be defined by
\begin{equation}
h_{i}(x)=\prod_{j=1}^{n} \mu_{j}^{\alpha_{ij}}(x_j),\ \ x
%=(x_1,\cdots,x_n)^{\tau}
\in \mathbbm{R}^n,
\ i\in \mathcal{S}.
\label{small012}
\end{equation}
One also has
$0 \leq h_{i}(x)\leq 1,\quad \sum_{i=1}^{s}h_{i}(x)=1.$
The It\^{o} stochastic  T-S model  can be
expressed by
\begin{equation}
dx\!=\!\sum_{i=1}^{s}h_{i}(x)\left(A_{i}x\! + \!B_{i}u \right)
dt\! + \!\sum_{i=1}^{s}h_{i}(x)C_{i}x dW(t).\label{a1}
\end{equation}
The premise variables in the  plant
rules described above are chosen as the state
variables $\{x_{j}\}_{j \in \mathcal{N}}$ which is necessary in
stochastic stability analysis and control design of the It\^{o}
stochastic T-S model in  (\ref{a1}). For an admissible control $u$ and an initial state
$x_0$, the trajectory $x(t)$ of  system (\ref{a1}) is
vector-valued stochastic processes which can be viewed as a mapping
$x=\left(x_{1},\ldots,x_{n}\right)^{\tau}: \mathbb{R}_{+} \times
\Omega \rightarrow \mathbb{R}^{n}$ with $
\mathbb{R}_{+}:=[0,\infty)$. Since the range of the  mapping (or
system state) $x$ is $ \mathbb{R}^{n}$, the antecedent parts that
$x_j$ is $ F_j^{\alpha_{ij}}$ in the  $i$th plant rule make
sense. Also $x(t)$ can be regarded  as the observation of
vector-valued stochastic processes as pointed out in
\cite{gfwq2013}.

\indent~~  This paper intends to develop  a line integral Lyapunov
function approach to stability analysis and stabilization for  the
It\^{o} stochastic  T-S model
(\ref{a1}). We first adopt the stochastic stability concept and the main
tool of stability analysis for very general It\^{o}  stochastic
dynamic systems. Consider an It\^{o} stochastic differential equation on $R^{n}$ of
the form
{\setlength\abovedisplayskip{6pt}
\setlength\belowdisplayskip{5pt}
\begin{equation}
dx=f(x,t)dt+g(x,t)dW(t)\label{b51}
\end{equation}}
\noindent where $f: \mathbb{R}^{n}\times \mathbb{R}_{+}\rightarrow
\mathbb{R}^{n} $, $g: \mathbb{R}^{n}\times \mathbb{R}_{+}
\rightarrow \mathbb{R}^{n}$ satisfy the usual linear growth and
local Lipschitz conditions for existence and uniqueness of solutions
to (\ref{b51}), $W(t)$ is a Wiener process adapted to a filtration
$\{\mathcal{F}_{t}\}_{t\geq0}$  that contains all $\mathcal{P}$-null sets and is right
continuous. Let $ x=0$ be the equilibrium point of (\ref{b51}),
that is, $ f(0,t)=0$, $g(0,t)=0$ and $x(t;t_{0},x_{0})$ be a
trajectory of  system (\ref{b51}) with
initial value $x(t_0)=x_0 \in \mathbbm{R}^n$.

%\textbf{Definition 1 }
 \begin{definition}
 \cite{XMao1997} The equilibrium
point $x=0 $ of equation (\ref{b51}) is said to be
\begin{description}
 \item[\hspace{0.15cm}i)] stochastically stable or stable in probability if for every pair of $\epsilon\in(0,1)$ and $r>0$, there exists a $\delta=\delta(\epsilon,r,t_{0})>0$ such that $\mathcal{P}\{|x(t;t_{0},x_{0})|<r, \ for \  all \  t\geq t_{0}\}\geq 1-\epsilon$ whenever $|x_{0}|<\delta$. Otherwise, it is said to be stochastically unstable.
\item[\hspace{0.15cm}ii)] stochastically asymptotically stable if it is stochastically stable and, moreover, for every $\epsilon\in(0,1)$, there exists a $\delta_0=\delta_0(\epsilon,t_{0})>0$ such that $\mathcal{P}\{\lim\limits_{t\rightarrow
\infty}x(t;t_{0},x_{0})=0\}\geq 1-\epsilon$ whenever
$|x_{0}|<\delta_{0}$.
\item[\hspace{0.15cm}iii)] stochastically asymptotically stable in the large if it is stable and, moreover,
for all $x_{0}\in \mathbb{R}^{n}$, $\mathcal{P}
\{\lim\limits_{t\rightarrow \infty}x(t;t_{0},x_{0})=0\}=1.$
\end{description}
\end{definition}
Now, we are in a position to introduce the line integral
Lyapunov function that was first proposed in \cite{RW2006}. A line integral  function   $V:R^n\rightarrow[0, \infty)$ is defined by
\begin{eqnarray}
V(x) \!=\!2\int_{\Gamma(0,x)}\bar{f}^{\tau}(\Psi) d\Psi\!:=\!2\int_{\Gamma(0,x)}\sum_{i\!=\!1}^{s}h_{i}(\Psi)\Psi^{\tau}P_{i} d\Psi\label{ea4}
\end{eqnarray}
\begin{eqnarray}
P_{i}&=&\bar{P}+D_{i}>0 \label{matrix}\\
\bar{P} & =& %\setlength{\arraycolsep}{1.5pt}
\left[
\begin{array}{cccc}
0 & p_{12} & \cdots & p_{1n} \\
p_{12} & 0 & \cdots & p_{2n} \\
\vdots & \vdots & \ddots & \vdots\\
p_{1n} & p_{2n} & \cdots & 0
\end{array}
\right] \label{a3-}\\
  D_{i}& =&%\setlength{\arraycolsep}{1pt}
\left[
\begin{array}{ccc}
d_{11}^{\alpha_{i1}} & \cdots & 0 \\
\vdots  & \ddots & \vdots \\
0  & \cdots & d_{nn}^{\alpha_{in}}
\end{array}
\right]\label{a3}
\end{eqnarray}
where $\Gamma(0,x)$ denotes a path from the origin to the current
state $x$; $\Psi\in\mathbb{R}^{n}$ is a dummy vector for the
integral and $d\Psi$ is an infinitesimal displacement vector.
To illustrate that $V(x)$ defined in (\ref{ea4})
 is a Lyapunov function, the following  fact is adopted from \cite{RW2006}.

\begin{fac} The line integral  $V(x)$ defined by (\ref{ea4}) is path-independent,  continuously differentiable, positive definite and also satisfies
\[V(x)\rightarrow \infty, \  \ ||x||\rightarrow \infty.\]
\end{fac}
\noindent It can be also seen from above fact and (\ref{ea4}) that the line
integral  Lyapunov function $V(x)$ is a
non-quadratic one.
The following simple  fact  reveals the
relationship between line integral Lyapunov function and
quadratic Lyapunov function.
\begin{fac} For the line integral Lyapunov function  $V(x)$ in
(\ref{ea4}), if there exists a positive definite matrix
$P$, such that the matrices $\{P_i\}_{i\in\mathcal{S}}$ satisfy
$P_{1}=\cdots=P_{s}=P$, then  $V(x)$ becomes
a  quadratic function $x^{\tau}Px$, that is, for any
$x\in \mathbbm{R}^n$
\[
V(x)=2\int_{\Gamma(0,x)}\Psi^{\tau}Pd\Psi=x^{\tau}Px.
\]
\end{fac}
For stability analysis of system (\ref{a1}) based on line
integral Lyapunov function $V(x)$ in (\ref{ea4}), the row gradient vector
${\partial
 V(x)}/{\partial x^{\tau}}$  is
needed. Due to the path independence of the line
integral, ${\partial
 V(x)}/{\partial x^{\tau}}$  can be calculated by the following Fact $3$.
\begin{fac} Consider the line integral Lyapunov function  $V(x)$ in
(\ref{ea4}). The row gradient vector
${\partial
 V(x)}/{\partial x^{\tau}}$ is calculated by
\begin{equation}
\frac{\partial
 V(x)}{\partial x^{\tau}}=2x^{\tau}P(x):=2\sum_{i=1}^{s}h_{i}(x)x^{\tau} P_{i}.  \label{zhzh}
\end{equation}
\end{fac}
The line integral Lyapunov function based stability analysis of the
  model in (\ref{a1}) is involved  in the
Hessian matrix $\frac{\partial^2
 V(x)}{\partial x \partial x^{\tau}}$ of  $V(x)$ in (\ref{ea4}) which is related to  partial derivatives of the fuzzy basis
 functions $ \{h_{i}(x)\}^s_{i=1}$. This also leads to the problem of handling
derivatives of membership functions (MFs) \cite{gb2012,KTanaka2003,BZhang2007}.  In what follows, we give a new method
to deal with the derivatives of MFs.
Consider the normalized membership functions
$\left\{\mu_j^{\alpha_{ij}}(x_j), (i, j)\in\mathcal{S}\times \mathcal{N}\right\}$ defined by (\ref{normalization}). Notice
that for a fixed pair $(i, j)\in\mathcal{S}\times \mathcal{N}$, the
function $\mu_j^{\alpha_{ij}}(x_j)$ is a scalar-valued one with only one
argument $x_j$. So we can give an assumption for the derivatives of
these functions. The assumption will form our starting point for
stochastic stability analysis of the model (\ref{a1}).
%\begin{align}
%\left|
%\begin{array}{c}
%   x_j\frac{d}{dx_j}\bigg(\mu_{j}^{\alpha_{ij}}(x_j)\bigg)
% \end{array}
% \right|
%\end{align}
%\textbf{Assumption 1}
\begin{assumption}
Assume for any $(i,j)\in
\mathcal{S}\times \mathcal{N}$ and $x_j \in \mathbb{R}$, that the
normalized membership function $\mu_{j}^{\alpha_{ij}}(x_j):
\mathbb{R}\mapsto [0,1]$ is differentiable and that there exists a
known  constant $\beta_{ij}>0$ satisfying
\begin{equation}
\bigg| x_j\frac{d}{dx_j}\left(\mu_{j}^{\alpha_{ij}}(x_j)\right)\bigg|\leq \beta_{ij}.
\label{mu-b}
\end{equation}
\end{assumption}
\begin{remark} It is important to point out that normalized membership functions
$\mu_j^{\alpha_{ij}}(x_j), i\in\mathcal{S}, j\in\mathcal{N}$
defined by (\ref{normalization}) are  state $x_j$
dependent and not the trajectory $x_j(t)$ dependent.  Therefore, the
derivative $\frac{d}{dx_j}(\mu_{j}^{\alpha_{ij}}(x_j))$ is  the
derivative of a scalar function  with respect to its one
argument $x_j$. If the membership functions are regarded as
trajectory dependent and their derivatives are  regarded as
time derivatives, the time derivatives of fuzzy basis functions will
be  related to the vector field that governs the T-S model. Such a case
makes the time derivatives of fuzzy basis functions   difficult
to handle \cite{gb2012}. So Assumption $1$ is different
from the ones in \cite {lk2014,KTanaka2003,BZhang2007} where
upper bounds of time derivatives of fuzzy basis functions are
required.
\end{remark}
\begin{remark} Assumption 1 implies that not only all the normalized membership functions $\mu_j^{\alpha_{ij}}(x_j) $ are differentiable, but also they should satisfy
 \begin{eqnarray*}
 -\frac{\beta_{ij}}{x_j}&\leq& \frac{d}{dx_j}(\mu_{j}^{\alpha_{ij}}(x_j))\leq  \frac{\beta_{ij}}{x_j},\ \  x_j>0\\
 \frac{\beta_{ij}}{x_j}&\leq&\frac{d}{dx_j}(\mu_{j}^{\alpha_{ij}}(x_j)) \leq  -\frac{\beta_{ij}}{x_j}, \ x_j<0.
\end{eqnarray*}
%\textcolor{blue}
%{\begin{equation*}
%\frac{d}{dx_j}(\mu_{j}^{\alpha_{ij}}(x_j))=
%\left\{
%\begin{array}{l}
%\begin{array}{c}
%    \leq  \frac{\beta_{ij}}{x_j}
%\end{array}
%,\ \  x_j>0 \\
%\begin{array}{cc}
%  \geq - \frac{\beta_{ij}}{-x_j}
%\end{array}, x_j<0
%  \end{array}
% \right.
%\end{equation*}}
%So if the membership functions of the involved fuzzy sets are not differentiable or do not satisfy the above requirements, the approach proposed in this paper cannot be applied. However,
Assumption $1$ for  the normalized membership function
$\mu_{j}^{\alpha_{ij}}(x_j)$ seems  to be reasonable and
the  upper bound of inequalities (\ref{mu-b}) can be fulfilled
for the exponential type of membership functions
\cite{zhou05a} (See  Example $2$  in Section $5$). Fig. 1 shows the region of $\frac{d}{dx_j}(\mu_{j}^{\alpha_{ij}}(x_j))$ satisfying Assumption $1$ with $\beta_{ij}=0.0125$.
\end{remark}
\vspace{-0.3cm}
\begin{figure}[H]
\begin{minipage}{7.2cm}
     \includegraphics[width=3.5in,height=1.6 in]{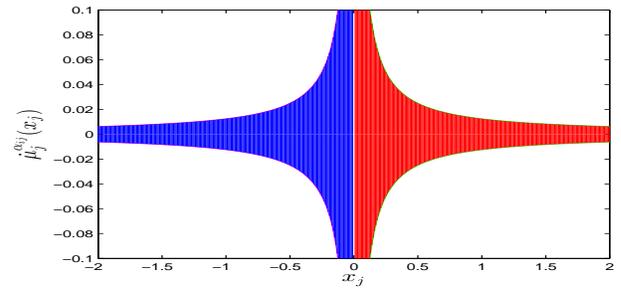}
\end{minipage}
\caption{Region of $\frac{d}{dx_j}(\mu_{j}^{\alpha_{ij}}(x_j))$ satisfying Assumption $1$}
\end{figure}
\vspace{-0.25cm}
It can be shown that the fuzzy basis functions $ \{h_{i}(x)\}^s_{i=1}$ are differentiable. An estimate
for all the partial derivatives of these fuzzy basis functions $\{h_{i}(x)\}^s_{i=1}$ can be derived by (\ref{mu-b}), and  will be
used later on.
\begin{lemma}\label{loyf1314}
Under Assumption $1$, the fuzzy basis functions $ \{h_{i}(x)\}^s_{i=1}$ given by
(\ref{small012}) are differentiable and satisfy
\begin{equation}
\left| x_{j}\frac{\partial
 h_{i}(x)}{\partial x_{j}}\right| \leq \beta_{ij},\ \ \ \ \ i\in \mathcal{S},\ \ \ j\in \mathcal{N}.\label{xjpar}
\end{equation}
\end{lemma}
\begin{proof}
By virtue of (\ref{small0}) and (\ref{small012}), one has
\begin{align}
&\left| x_{j}\frac{\partial h_{i}(x)}{\partial x_{j}}\right|=\bigg| x_{j}\frac{\partial \left(\prod\limits_{k=1}^{n} \mu_{k}^{\alpha_{ik}}(x_k)\right)}{\partial x_{j}}\bigg|= \nonumber\\
&\bigg| x_{j}\prod\limits_{k\neq j}^{n} \mu_{k}^{\alpha_{ik}}(x_k)\frac{d}{dx_j}\left(\mu_{j}^{\alpha_{ij}}(x_j)\right)\bigg| \leq \left| x_{j}\frac{d}{dx_j}\left(\mu_{j}^{\alpha_{ij}}(x_j)\right)\right|   \nonumber
\end{align}
 which together with (\ref{mu-b}) leads to the inequalities (\ref{xjpar}).
\end{proof}
Now, we give the new results of Hessian matrix for the line integral Lyapunov function $V(x)$ in (\ref{ea4}) which will be adopted in the following section.
\begin{lemma}\label{loyf13141}
Consider the fuzzy basis functions $ \{h_{i}(x)\}^s_{i=1}$ in
(\ref{small012}) with the normalized membership functions $\mu_{j}^{\alpha_{ij}}(x_j)$ satisfying (\ref{mu-b}) in Assumption $1$.  For matrices $D_i>0$ defined in (\ref{a3}), a positive definite matrix $D$  with $D-D_{i}\geq 0$ and $P(x)$ defined in (\ref{zhzh}),   the following inequalities with respect to Hessian matrix hold for $x, y \in \mathbb{R}^{n}$
\begin{eqnarray}
y^\tau\bigg(\frac{\partial^2
 	V(x)}{\partial x \partial x^{\tau}}\bigg) y&=&y^\tau\bigg(P(x)\!+\!\sum\limits^s_{i=1}\frac{\partial h_{i}(x)}{\partial x}x^{\tau}D_i\bigg)y\nonumber\\
&\leq& y^\tau\bigg(P(x)\!+\!\beta D\bigg) y \label{ypath2}\\
 \beta&=&\sum\limits^n_{j=1}\sum\limits^s_{i=1}\beta_{ij} \label{bound2}
\end{eqnarray}
with the known constants $\beta_{ij}>0$ given in (\ref{mu-b}).
\end{lemma}
 \emph{\text{ \indent Proof:\ }}
For the $V(x)$ defined in (\ref{ea4}),  in terms of
Fact $3$ We have
\begin{equation}
\frac{\partial
	V(x)}{\partial t}=0,\ \ \ \ \ \ \  \frac{\partial
	V(x)}{\partial x^{\tau}}=2x^{\tau}P(x)\label{aa3}.
\end{equation}
%Letting $D=diag\{d_1,\ \cdots,\ d_n\}$ and
%using the condition $D-D_{j}\geq 0$, we can
%conclude that $d_{jj}^{\alpha_{ij}}\leq d_j$.
On the other hand, recalling (\ref{matrix})-(\ref{a3}) gives
\begin{align}
&2x^{\tau}P(x)=2\left(x_{1},\ldots,x_{j},\ldots,x_{n}\right)\times \nonumber\\
%\end{align}
%\begin{align}
&\left[\begin{array}{ccccccc}
\sum\limits_{i=1}^{s}h_{i}d_{11}^{\alpha_{i1}} & \cdots & p_{1j} & \cdots & p_{1n} \\
\vdots &  &\vdots &  &\vdots\\
p_{j1} & \cdots  & \sum\limits_{i=1}^{s}h_{i}d_{jj}^{\alpha_{ij}} & \cdots  & p_{jn}\\
\vdots &  & \vdots &  & \vdots\\
p_{n1} & \cdots  & p_{nj} & \cdots &
\sum\limits_{i=1}^{s}h_{i}d_{nn}^{\alpha_{in}}
\end{array}
\right]\nonumber\\
%\end{align}
%\begin{align}
&=2\bigg(x_{1}\sum_{i=1}^{s}h_{i}d_{11}^{\alpha_{i1}}\!+\!\sum_{k=2}^{n}x_{k}p_{k1},
\ldots , \sum_{k\neq
	j}x_{k}p_{kj}+x_{j} \sum_{i=1}^{s} \nonumber\\
%\end{align}
%\begin{align}
&h_{i} d_{jj}^{\alpha_{ij}},\cdots
,\sum_{k=1}^{n-1} x_{k}p_{kn}\!+\! x_{n} \sum_{i=1}^{s} h_{i}d_{nn}^{\alpha_{in}}\bigg).\label{1111w}
\end{align}
By (\ref{1111w}), the Hessian matrix of the line integral Lyapunov function $V(x)$ can be calculated by
 \begin{align}
 &\frac{\partial^2
 	V(x)}{\partial x \partial x^{\tau}}=\frac{\partial
 }{\partial x}\left[\frac{\partial
 	V(x)}{\partial x^{\tau}}\right]=
 2\frac{\partial
 	\left[x^{\tau}P(x)\right]}{\partial x}\nonumber\\
 & =2\partial\bigg(x_{1}\sum_{i=1}^{s}h_{i}d_{11}^{\alpha_{i1}}\!+\!\sum_{k=2}^{n}x_{k}p_{k1},
 \ldots , \sum_{k\neq
 	j}x_{k}p_{kj}+ x_{j}\times \nonumber
 \end{align}
 \begin{align}
 &\ \ \ \ \  \sum_{i=1}^{s}h_{i}d_{jj}^{\alpha_{ij}},\cdots
 ,\sum_{k=1}^{n-1} x_{k}p_{kn}\!+\! x_{n} \sum_{i=1}^{s} h_{i}d_{nn}^{\alpha_{in}}\bigg)/{\partial x}
 \nonumber\\
 &=2\left[\begin{array}{c|c|c}
 \sum\limits_{i=1}^{s}h_{i}d_{11}^{\alpha_{i1}}+x_1 \sum\limits_{i=1}^{s}\frac{\partial h_{i}}{\partial x_1}d_{11}^{\alpha_{i1}}
 & \cdots & \  \\ \hline
 \vdots & \ddots & \   \\ \hline
 p_{n1}+x_1 \sum\limits_{i=1}^{s}\frac{\partial h_{i}}{\partial x_n}d_{11}^{\alpha_{i1}} & \cdots  & \
 \end{array} \right.\nonumber \\
 & \ \ \ \ \ \ \ \ \ \ \ \ \ \ \ \  \left.\begin{array}{c|c}
 \ &  p_{1n}+x_n \sum\limits_{i=1}^{s}\frac{\partial h_{i}}{\partial x_1}d_{nn}^{\alpha_{in}} \\ \hline
 \ &   \vdots \\ \hline
 \ &   \sum\limits_{i=1}^{s}h_{i}d_{nn}^{\alpha_{in}}+x_n \sum\limits_{i=1}^{s}\frac{\partial h_{i}}{\partial x_n}d_{nn}^{\alpha_{in}}
 \end{array} \right]\label{xnnn}\\
% \end{align}
% \begin{align}
 &=2 \times \left\{ \left[\begin{array}{ccccccc}
 \sum\limits_{i=1}^{s}h_{i}d_{11}^{\alpha_{i1}} %& \cdots & p_{1j}
 & \cdots & p_{1n} \\
 \vdots & \ddots & \vdots\\
 %p_{j1} & \cdots & \sum_{i=1}^{s}h_{i}(x)d_{jj}^{\alpha_{ij}} & \cdots & p_{jn}\\
 %\vdots & & \vdots
 %& & \vdots\\
 p_{n1} %& \cdots & p_{nj}
 & \cdots &
 \sum\limits_{i=1}^{s}h_{i}d_{nn}^{\alpha_{in}}
 \end{array}
 \right]\right.\nonumber\\
 &\left.\ \ \ \ \ \ \ \ \ \ +\left[\begin{array}{ccccccc}
 x_1 \sum\limits_{i=1}^{s}\frac{\partial h_{i}}{\partial x_1}d_{11}^{\alpha_{i1}}& \cdots %& x_j \sum_{i=1}^{s}\frac{\partial h_{i}}{\partial x_1}d_{jj}^{\alpha_{ij}} & \cdots
 & x_n\sum\limits_{i=1}^{s}\frac{\partial h_{i}}{\partial x_1}d_{nn}^{\alpha_{in}} \\
 \vdots  &\ddots & \vdots\\
 %x_1 \sum_{i=1}^{s}\frac{\partial h_{i}}{\partial x_j}d_{11}^{\alpha_{i1}} & \cdots %& x_j \sum_{i=1}^{s}\frac{\partial h_{i}}{\partial x_j}d_{jj}^{\alpha_{ij}}& \cdots
 %&x_n \sum_{i=1}^{s}\frac{\partial h_{i}}{\partial x_j}d_{nn}^{\alpha_{in}}\\
 %\vdots & & \vdots & & \vdots\\
 x_1 \sum\limits_{i=1}^{s}\frac{\partial h_{i}}{\partial x_n}d_{11}^{\alpha_{i1}} & \cdots %& x_ j \sum_{i=1}^{s}\frac{\partial h_{i}}{\partial x_n}d_{jj}^{\alpha_{ij}} & \cdots
 & x_n \sum\limits_{i=1}^{s}\frac{\partial h_{i}}{\partial x_n}d_{nn}^{\alpha_{in}}
 \end{array}
 \right] \right\}\nonumber\\
 &=2\bigg(P(x)+\sum_{i=1}^{s}\frac{\partial h_{i}}{\partial x}x^{\tau}D_i\bigg).
 % &=2P(x)+2\sum_{i=1}^{s}\frac{\partial h_i}{\partial x}x^{\tau}D_i
 \label{aa4}
 \end{align}
In order to achieve an upper bound of the Hessian matrix, the upper bound of the second part in (\ref{aa4}) is needed to be estimated. For any $y \in \mathbb{R}^{n}$, it is clear by $D_i>0$ that
\begin{align}
y^\tau\frac{\partial h_{i}}{\partial x}x^{\tau}D_iy\!=\!y^\tau D_i^{\frac{1}{2}} D_i^{\!-\!\frac{1}{2}}\frac{\partial h_{i}}{\partial x}x^{\tau}D_i^{\frac{1}{2}}D_i^{\frac{1}{2}} y.\label{rank1}
\end{align}
In the derivation above, we drop the argument of $h_i(x)$ for
simplicity. In what follows, we will do so in some cases for
notational convenience. Since
\begin{align}
rank\bigg(D_i^{\!-\!\frac{1}{2}}\frac{\partial h_{i}}{\partial x}x^{\tau}D_i^{\frac{1}{2}}\bigg)\leq rank(x^\tau) \leq 1 \label{rank+}
\end{align}
  there exist an invertible matrix $P$ such that
 \begin{align}
 &P\bigg(D_i^{\!-\!\frac{1}{2}}\frac{\partial h_{i}}{\partial x}x^{\tau}D_i^{\frac{1}{2}}\bigg)P^{-1}=\nonumber\\
 &diag\bigg\{tr\bigg(PD_i^{\!-\!\frac{1}{2}}\frac{\partial h_{i}}{\partial x}x^{\tau}D_i^{\frac{1}{2}}P^{\!-\!1}\bigg), 0, \cdots, 0\bigg\}.\label{pdfra}
 \end{align}
 It  follows from (\ref{rank1}) and (\ref{pdfra}) that
\begin{align}
&y^\tau\left(\frac{\partial h_{i}}{\partial x}x^{\tau}D_i\right)y\!=\!\bigg(P^{\!-\!\tau}D_i^{\frac{1}{2}}y\bigg)^\tau \bigg(PD_i^{\!-\!\frac{1}{2}}\frac{\partial h_{i}}{\partial x}x^{\tau}D_i^{\frac{1}{2}}P^{\!-\!1}\bigg )\nonumber\\
&\!\times\!\bigg(PD_i^{\frac{1}{2}} y\bigg)\!=\!\bigg(P^{\!-\!\tau}D_i^{\frac{1}{2}}y\bigg)^\tau diag\bigg\{tr\bigg(PD_i^{\!-\!\frac{1}{2}}\frac{\partial h_{i}}{\partial x}x^{\tau}D_i^{\frac{1}{2}}P^{\!-\!1}\bigg)\nonumber \\
& ,0, \cdots, 0\bigg\}\bigg(PD_i^{\frac{1}{2}} y\bigg)\leq \left| tr\bigg(PD_i^{-\frac{1}{2}}\frac{\partial h_{i}}{\partial x}x^{\tau}D_i^{\frac{1}{2}}P^{\!-\!1}\bigg)\right|\nonumber\\
& \bigg(P^{-\tau}D_i^{\frac{1}{2}}y\bigg)^\tau \bigg(PD_i^{\frac{1}{2}} y\bigg).\label{fracbi}
\end{align}
In light of the operational property of matrix trace, we have
\begin{align}
tr(PD_i^{\!-\!\frac{1}{2}}\frac{\partial h_{i}}{\partial x}x^{\tau}D_i^{\frac{1}{2}}P^{\!-\!1})\!=\!tr\left(\frac{\partial h_{i}}{\partial x}x^{\tau}\right)\!=\! x^\tau\frac{\partial h_{i}}{\partial x}.\label{hxd1}
\end{align}
Combining (\ref{fracbi}) and (\ref{hxd1}) yield
\begin{align}
&y^\tau\left(\frac{\partial h_{i}}{\partial x}x^{\tau}D_i\right)y\!\leq\!\left| tr\bigg(\frac{\partial h_{i}}{\partial x}x^{\tau}\bigg)\right|\bigg(P^{-\tau}D_i^{\frac{1}{2}}y\bigg)^\tau\times \nonumber\\
&\bigg(PD_i^{\frac{1}{2}} y\bigg)=\bigg|\sum_{j=1}^{n}x_j\frac{\partial h_{i}}{\partial x_j}\bigg| y^\tau D_iy\leq\sum_{j=1}^{n}\left| x_j\frac{\partial h_{i}}{\partial x_j}\right|y^\tau D_iy\nonumber.
\end{align}
In terms of (\ref{xjpar}) in Lemma $1$ and  $D-D_{i}\geq 0$, we obtain
\begin{eqnarray}
y^\tau\frac{\partial h_{i}}{\partial x}x^{\tau}D_iy
 &\leq& \sum\limits_{j=1}^{n}\beta_{ij} y^\tau D_iy\nonumber\\
 &\leq& \sum\limits_{j=1}^{n}\beta_{ij} y^\tau Dy\nonumber
%=y^\tau \beta D_iy\leq y^\tau \beta D y
\end{eqnarray}
which together with (\ref{bound2}) gives that
\begin{eqnarray*}
y^{\tau}\left(\sum\limits_{i\!=\!1}^{s}\frac{\partial h_{i}}{\partial x}x^{\tau}D_i\right)y &\leq& \sum\limits_{i\!=\!1}^{s}\sum\limits_{j\!=\!1}^{n}\beta_{ij} y^\tau Dy\\
&\leq& \beta y^\tau  D y.
\end{eqnarray*}
Therefore, the inequality (\ref{ypath2}) can be obtained easily.
%\noindent This result can directly be checked by using Assumption $1$ and
%(\ref{small012}).
 %Now choose a positive definite matrix $\Upsilon$ such that
%\begin{equation}
%\Upsilon \geq \left[
%\begin{array}{cccc}
% \sum_{i=1}^{s}\beta_{i1}& \cdots & 0 \\
%\vdots & \ddots & \vdots \\
%0 & \cdots & \sum_{i=1}^{s}\beta_{in}
%\end{array}
%\right]>0. \label{up-b}
%\end{equation}
\section{Stochastic Stability Analysis}
In this section, relaxed stability conditions for the unforced
It\^{o} stochastic T-S model will be given by using the
line integral Lyapunov function $V(x)$. The unforced stochastic T-S
model   can be presented by
\begin{equation}
dx= \sum_{i=1}^{s}h_{i}A_{i}x dt+\sum_{i=1}^{s}h_{i}C_{i}x
dW(t)\label{d1}
\end{equation}
whose stability conditions are presented by Theorem $1$.
\begin{theorem} \label{wnai1*}
Under Assumption $1$, the equilibrium point of the unforced It\^{o} stochastic T-S
model based system in  (\ref{d1}) is stochastically
asymptotically stable in the large, if there exist matrices $\{D_{j}\}_{j\in\mathcal{S}}$ and $\bar{P}$
being of the forms given in (\ref{a3}),  a  diagonal matrix $D$  as well as matrices  $\{Q_{ij}:
Q_{ij}^{\tau}=Q_{ij}=Q_{ji}\}_{i,j \in \mathcal{S}}$ such that for all $i,j \in \mathcal{S}$,  the
inequalities hold:
\begin{align}
& P_{j}=\bar{P}+D_{j}>0, \ D-D_{j}\geq 0 \label{*e1h+}\\
%\end{align}
%\begin{align}
&(P_{j}A_{i})^S+C_{i}^{\tau}(P_{j}+\beta D)C_{i}+Q_{ij}<0\label{*2c61++} \\
&\Theta_1=\left[
\begin{array}{ccccc}
Q_{11} & \ldots & Q_{1s}\\
\vdots & \ddots & \vdots\\
Q_{s1} & \ldots & Q_{ss}
\end{array}
\right] > 0\label{*2c6+}
\end{align}
where $\beta$ is given in (\ref{bound2}).
\end{theorem}
\begin{proof} Suppose that there exist matrices $\{D_{j}\}_{j
\in\mathcal{S}}$, $\bar{P}$ and $D$ as well as matrices  $\{Q_{ij}
\}_{i,j \in \mathcal{S}}$ such that the matrix
inequalities (\ref{*e1h+})-(\ref{*2c6+}) hold. Using matrices
$\{P_{j}=\bar{P}+D_{j}\}_{j\in\mathcal{S}}$ and the fuzzy basis
functions $\{h_{j}\}_{j\in\mathcal{S}}$, one can construct a
vector filed $\bar{f}(x)$ on the state space $\mathbb{R}^{n}$ by
{\setlength\abovedisplayskip{6pt}
\setlength\belowdisplayskip{5pt}
\begin{equation}
\bar{f}(x)=P(x)x=\Sigma_{j=1}^{s}h_{j}P_{j}x. \label{zh1-}
\end{equation}}
Given a fixed state $x\in\mathbbm{R}^n$ arbitrarily, for a path
$\Gamma(0,x)$ from the origin to $x$, let
{\setlength\abovedisplayskip{6pt}
\setlength\belowdisplayskip{5pt}
\begin{equation} V(x)=2\int_{\Gamma(0,x)}\bar{f}^{\tau}(\Psi)
d\Psi. \label{zh1}
\end{equation}}
One can see from Fact $1$ that $V(x)$ is a Lyapunov function candidate.
It follows from  (\ref{aa3}),  (\ref{aa4}) in Lemma $2$ and (\ref{d1}),
that
\begin{align}
&\mathcal {L} V(x)=\frac{\partial V(x)}{\partial t}+\frac{\partial
 V(x)}{\partial x^{\tau}}\left(\sum_{i=1}^{s}h_{i}A_{i}x\right)\nonumber\\
 &+\frac{1}{2}tr\left[\left(\sum_{i=1}^{s}h_{i}C_{i}x\right)^{\tau}\frac{\partial^2
      V(x)}{\partial x \partial x^{\tau}}\left(\sum_{i=1}^{s}h_{i}C_{i}x\right)\right]\nonumber\\
 &=2x^{\tau}P(x)\left(\sum_{i=1}^{s}h_{i}A_{i}x\right)+\frac{1}{2}tr\left[\left(\sum_{i=1}^{s}h_{i}C_{i}x\right)^{\tau} \right.\nonumber\\
 &\times\left.2\bigg[P(x)+\sum_{i=1}^{s}\frac{\partial h_{i}}{\partial x}x^{\tau}D_i\bigg]\left(\sum_{i=1}^{s}h_{i}C_{i}x\right)\right].\label{aa1}
   %& =\sum_{i=1}^{s}\sum_{j=1}^{s}\sum_{k=1}^{s}h_{i}h_{j}h_{k}x^{\tau}\left\{
%      P_{j}A_{i}+A_{i}^{\tau}P_{j}+\frac{1}{2}[C_{i}^{\tau}\right.\nonumber\\
% & \left. \times\left(P_{j}+\beta D\right)C_{k}
%    \!+\!C_{k}^{\tau}\left(P_{j}\!+\!\beta D\right)C_{i}] \right \}x.
\end{align}
Choosing $y=\left(\sum\limits_{i=1}^{s}h_{i}C_{i}x\right)$ in Lemma $2$, in terms of (\ref{aa1}),  one has
%By using Lemma \ref{loyf13141}, one can see that
%\begin{align}
%&\left(\sum_{i=1}^{s}h_{i}C_{i}x\right)^{\tau}\left(\sum_{i=1}^{s}\frac{\partial h_{i}}{\partial x}x^{\tau}D_i\right)\left(\sum_{i=1}^{s}h_{i}C_{i}x\right)\nonumber
%\end{align}
%\begin{align}
%&\leq\left(\sum_{i=1}^{s}h_{i}C_{i}x\right)^{\tau}\left(\sum_{i=1}^{s}\sum_{j=1}^{n}\beta_{ij}D\right)\left(\sum_{i=1}^{s}h_{i}C_{i}x\right)\nonumber\\
%&=\left(\sum_{i=1}^{s}h_{i}C_{i}x\right)^{\tau}\left(\beta D\right)\left(\sum_{i=1}^{s}h_{i}C_{i}x\right).\label{left1}
%\end{align}
\begin{align}
&\mathcal {L} V(x)\leq2x^{\tau}P(x)\left(\sum_{i=1}^{s}h_{i}A_{i}x\right)+\frac{1}{2}tr\left[\left(\sum_{i=1}^{s}h_{i}C_{i}x\right)^{\tau} \right.\nonumber\\
 &\times\left.\bigg(2P(x)
  +2\beta D\bigg)\left(\sum_{i=1}^{s}h_{i}C_{i}x\right)\right]
  =\sum_{i=1}^{s}\sum_{j=1}^{s}\sum_{k=1}^{s}h_{i}h_{j}h_{k}\nonumber \\
  &\times x^{\tau}\left[P_{j}A_{i}+\frac{1}{2}C_{i}^\tau\left(P_{j}+\beta D\right)C_{k}\right]^Sx.\label{aa101}
\end{align}
 Fo $i,j,k\in\mathcal{S}$, inequalities
(\ref{*e1h+}) and $\beta \mathcal{D} >0$ lead to
\[[C_{i}^{\tau}\left(P_{j}\!+\!\beta D\right)C_{k}]^S
\!\leq\! C_{i}^{\tau}\left(P_{j}\!+\!\beta D\right)C_{i}\!+\!C_{k}^{\tau}\left(P_{j}\!+\!\beta D\right)C_{k}. \]
Using these inequalities and  (\ref{aa1}), one has that
\begin{eqnarray}
\mathcal {L} V(x) \!\leq\! \sum_{i\!=\!1}^{s}\sum_{j\!=\!1}^{s}h_{i}h_{j}x^{\tau}\left[(P_{j}A_{i})^S\!+\!C_{i}^{\tau}(P_{j}\!+\!\beta D)C_{i}\right]x\label{a2}
\end{eqnarray}
which together with (\ref{*2c61++}) gives
\begin{align*}
\mathcal {L} V\!<-\!x^{\tau}
\left[\begin{array}{ccc}
 h_{1}I & \cdots
& h_{s}I
\end{array}
\right]
\Theta_1
\left[
\begin{array}{c}
h_{1}I \\ \vdots \\
h_{s}I
\end{array}
\right]x. \label{c61-}
\end{align*}
Then it follows from (\ref{*2c6+}) that for $x\neq 0$,
$\mathcal {L} V< 0 $.
Recalling that $V(x)$ is
radially unbounded, we can conclude that the equilibrium point of
system (\ref{d1}) is stochastically asymptotically stable in the
large by Lemma $3$ in \cite{zhouren2010}.
\end{proof}
\begin{remark}
For the stability analysis of the unforced It\^{o} stochastic T-S fuzzy system (\ref{d1}), the involved Hessian
matrix $\frac{\partial^2
V(x)}{\partial x \partial x^{\tau}}$ of the line integral  $V(x)$  is related to  partial derivatives of the  basis
 functions. It is clear from (\ref{xnnn}) that the entry $\mathscr{H}_{ij}$ in slot $(i,j)$ for Hessian
matrix is given by
 \begin{align}
 \mathscr{H}_{ij}=
\left\{
\begin{array}{l}
\begin{array}{cc}
   \sum\limits_{i=1}^{s}h_{i}d_{jj}^{\alpha_{ij}}+x_j\sum\limits_{i=1}^{s}\frac{\partial h_{i}}{\partial x_j}d_{jj}^{\alpha_{ij}},
\end{array}
  i=j \\
\begin{array}{cc}
  P_{ij}+x_j\sum\limits_{i=1}^{s}\frac{\partial h_{i}}{\partial x_j}d_{jj}^{\alpha_{ij}},
\end{array}  \ \ \ \  \ \ \ \  i\neq j .
  \end{array}
 \right.\nonumber
 \end{align}
 So the Hessian
matrix is highly dependent on the state $x$ and the partial derivatives of the  basis
 functions, which gives rise to difficulty and challenging for estimating the upper bound of $\mathcal {L}V(x)$.  Fortunately, this Hessian matrix can be decomposed into two terms in (\ref{aa4}). The key point is to handle the second term $\sum_{i=1}^{s}\frac{\partial h_{i}}{\partial x}x^{\tau}D_i$ which is involved in (\ref{aa1}) of $\mathcal {L}V$.
 \end{remark}
 \begin{remark}
 To our best knowledge, there
is no results on  dealing with the Hessian matrix of line integral  function, therefore we give a novel  Lemma $2$ to do so. Facilitating by Assumption $1$, an upper bound $ y^\tau\beta Dy $ of $y^\tau\sum^s_{i\!=\!1}\frac{\partial h_{i}}{\partial x}x^{\tau}D_iy$ can be obtained.
The idea behind this lemma  is as follows. It is observed that the rank of matrix $D_i^{\!-\!\frac{1}{2}}\frac{\partial h_{i}}{\partial x}x^{\tau}D_i^{\frac{1}{2}}$ is less than or equal to $1$. So this matrix is equivalent to a diagonal matrix with only one nonzero element $tr(PD_i^{\!-\!\frac{1}{2}}\frac{\partial h_{i}}{\partial x}x^{\tau}D_i^{\frac{1}{2}}P^{\!-\!1})\!=\! x^\tau\frac{\partial h_{i}}{\partial x}$. This elegant property of matrix with rank one together with Lemma $1$ leads to the upper bound $\beta y^\tau  D y$.
Based on Lemma $2$, an estimate of  $\mathcal {L}V$  can be achieved which plays an important role on the stability analysis of the  unforced stochastic T-S
fuzzy system (\ref{d1}).
\end{remark}
 Similarly,
 the following corollary based on the common quadratic
Lyapunov function can be obtained.

\begin{corollary}$\label{xisf12}$
If there exist positive matrices  $\{Q_{i}\}_{i \in \mathcal{S}}>0$
and  $P>0$ satisfying inequalities
{\setlength\abovedisplayskip{6pt}
\setlength\belowdisplayskip{5pt}
\begin{eqnarray}
(PA_{i})^S+C_{i}^{\tau}PC_{i}+Q_{i}<0,\quad i \in
\mathcal{S}\label{c61+-}
\end{eqnarray}}
then the equilibrium point of the unforced It\^{o} stochastic T-S
model (\ref{d1}) is stochastically
asymptotically stable in the large.
\end{corollary}

\begin{remark}%The conditions for stochastic asymptotic stability of the unforced
%It\^{o} stochastic T-S model   (\ref{d1}) in
%Theorem \ref{wnai1*} are obtained via a technique which is
%similar to those in \cite{k2000} and \cite{zhou07}.
It is noted that  for any $i, j \in \mathcal{S}$ if
 $P_i=P,\ Q_{ii}=Q_{i}$
 and $ Q_{ij}=0, for\ i\neq j$ in inequalities
(\ref{*e1h+})-(\ref{*2c6+}),
 the conditions in Theorem \ref{wnai1*}
 become conditions in
Corollary $1$. This means that if inequalities
(\ref{c61+-}) have a set of solutions  $\{Q_{i}\}_{i\in
\mathcal{S}}$ and $P$, then $\{P_{j}=P\}_{j\in \mathcal{S}}$,
$\{Q_{ii}=Q_{i}\}_ {i\in \mathcal{S}}$ and $\{Q_{ij}=0,\ i\neq j,\ i,
j\in \mathcal{S} \}$ are solutions to inequalities
  (\ref{*e1h+})-(\ref{*2c6+}). However, if the inequalities
 (\ref{*e1h+})-(\ref{*2c6+}) have a set of solutions $\{P_{j}\}_{j\in
  \mathcal{S}}$, $\{Q_{ij}\}_
  {i, j\in \mathcal{S}}$, the inequalities (\ref{c61+-}) do
  not necessarily have a solution (See Example
  $1$). Thus, the line integral based stochastic stability result is more general than that derived from
quadratic Lyapunov function.
\end{remark}
 The following example is given to illustrate that the line integral based analysis results of stochastic
asymptotic stability in Theorem $1$ are less conservative than the ones based on quadratic Lyapunov function.\vspace{0.15cm}
%\noindent \textbf{Example} $1:$\quad
\begin{example}
For the  unforced stochastic T-S
model (\ref{d1}), the fuzzy rules are given as follows:\\
%\par\setlength\parindent{3em}
\textbf{$R_{1}$}: $ $If $x_{1}$ is $ F_{1}^{1}$ and $x_{2}$ is $
F_{2}^{1}$, then $dx=A_{1}xdt+C_{1}x dW(t)$\\% \par\setlength\parindent{1.5em}
\textbf{$R_{2}$}: $ $If $x_{1}$ is $ F_{1}^{1}$ and $x_{2}$ is $
F_{2}^{2}$,  then $dx=A_{2}xdt+C_{2}x dW(t)$
\textbf{$R_{3}$}:  $ $If $x_{1}$ is $ F_{1}^{2}$ and $x_{2}$ is $
F_{2}^{1}$, then   $dx=A_{3}xdt+C_{3}x dW(t)$
 \textbf{$R_{4}$}: $ $If $x_{1}$ is $ F_{1}^{2}$ and $x_{2}$ is $
F_{2}^{2}$, then  $dx=A_{4}xdt+C_{4}x dW(t)$
\vspace{-0.2cm}
\begin{alignat*}{4}
&A_{1}\!=\!\left[
\begin{array}{cc}
  \!-\!0.5 & 0\\
   \!-\!1 & a
\end{array} \right]\!,\!
A_{2}\!=\!\left[
\begin{array}{cc}
 \!-\!0.8 & 1\\
  0 & \!-\!1
\end{array} \right]\!,\!
 A_{3}\!=\!\left[
\begin{array}{cc}
  \!-\!0.9 & 1.09 \\
         0 & \!-\!0.8
\end{array} \right]\\
&A_{4}\!=\!\left[
\begin{array}{cc}
  b & 1 \\
         0 & \!-\!0.4
\end{array} \right]\!,\!
C_{1}=\left[
\begin{array}{cc}
  0.2 & 0 \\
  0 & \!-\!0.1
\end{array} \right]\!,\!
C_{2}\!=\!\left[
\begin{array}{cc}
  0.3 & 0 \\
         0 & \!-\!0.2
\end{array} \right]\\
&C_{3}=\left[
\begin{array}{cc}
  0.2 & 0 \\
         0 & \!-\!0.3
\end{array} \right],
C_{4}=\left[
\begin{array}{cc}
  0.35 & 0 \\
         0 & \!-\!0.1
\end{array} \right]
\end{alignat*}
where $a$ and $b$ are real parameters.
Then all ordinal numbers $\alpha_{ij}, i=1, \cdots,4; j=1,2$ read as
\begin{align*}
  &\alpha_{11}=1,\  \alpha_{21}=1,\  \alpha_{31}=2,\  \alpha_{41}=2\\
  &\alpha_{12}=1,\  \alpha_{22}=2,\  \alpha_{32}=1,\  \alpha_{42}=2
\end{align*}
with $s_1=s_2=2$.
\vspace{-0.3cm}
\begin{figure}[H]
\begin{minipage}{7.2cm}
     \includegraphics[width=3.3in,height=1.5 in]{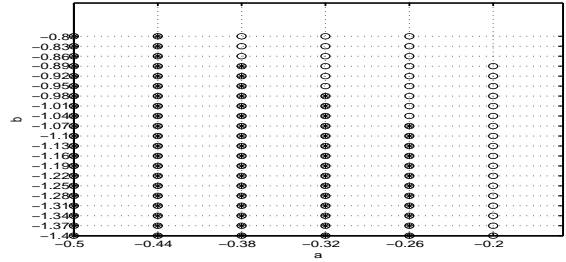}
\end{minipage}
\caption{Stable  regions: $\circ$---Theorem 1  and $\star$---Corollary 1}
\end{figure}
\vspace{-0.3cm}
%\begin{figure}[H]
%\setlength{\abovecaptionskip}{-0.25cm}
%\begin{center}
   % \includegraphics[width=3.3in,height=2.0in]{sto.eps}
    %\captionsetup[figure]{indention=5mm}
  %  \caption{Stochastically
%asymptotically stable  regions:$\circ$---Theorem 1  and $\star$---Corollary 1}
 %   \end{center}
%\end{figure}
\noindent The figure above shows the regions of stochastic asymptotic stability. The points marked with $``\circ$" and $``\star"$ show the regions of stochastic asymptotic stablity  obtained by using Theorem 1 and Corollary 1, respectively. %From Fig.1,  it is clear that the results in Theorem $1$ based on the line integral Lyapunov approach  yield larger  regions of asymptotic stability and extended dissipativity than the ones in Corollary $1$ based on quadratic Lyapunov method.
From Fig. $2$, it is clear that the results in Theorem $1$ yield larger  regions of stochastic stability than the ones in Corollary $1$.
As expected, stability results derived from
line integral  function are less conservative than those derived from
quadratic Lyapunov  function.
\end{example}

\section{Fuzzy Controller Design}
\noindent In this section, stability analysis and controller design for
It\^{o} closed-loop stochastic  system will be addressed. More
precisely, we are interested in finding a fuzzy controller such that the resulting
It\^{o} closed-loop stochastic system is stochastically
asymptotically stable. The fuzzy control law can be described as
follows:\\
\textbf{Controller Rule} $ i$:\\
 \textbf{If} $x_{1}$ is $%
F_{1}^{\alpha_{i1}},$ and $\ldots $ , $x _{n}$ is
$F_{n}^{\alpha_{in}}$,  \textbf{then}
$u=K_{i}x,  \ i\in \mathcal{S}.$
 The overall state feedback fuzzy controller is
represented by
{\setlength\abovedisplayskip{6pt}
\setlength\belowdisplayskip{5pt}
\begin{equation}
u=\sum_{i=1}^{s}h_{i}K_{i}x.  \label{d2}
\end{equation}}
Substituting (\ref{d2}) into (\ref{a1}) yields
{\setlength\abovedisplayskip{6pt}
\setlength\belowdisplayskip{5pt}
\begin{equation}
dx= \sum_{i=1}^{s}\sum_{j=1}^{s}h_{i}h_{j}A_{ij}x
dt+\sum_{i=1}^{s}h_{i}C_{i}x dW(t)\label{e6}
\end{equation}}
where
{\setlength\abovedisplayskip{0.5pt}
\setlength\belowdisplayskip{0.5pt}
\begin{equation}
A_{ij}=A_{i} + B_{i}K_{j}.\label{e61}
\end{equation}}
For the closed-loop It\^{o}  stochastic system (\ref{e6}), we have
the following result.
\begin{theorem} \label{the+*-}
Under Assumption $1$, the equilibrium point of the  closed-loop It\^{o} stochastic
system (\ref{e6}) is stochastically asymptotically stable in
the large, if there exist matrices $\{D_{k}\}_{k\in\mathcal{S}}$ and $\bar{P}$
being of the forms given in (\ref{a3}),  a  diagonal matrix $D$  , as well as
matrices $\{\bar Q_{ijk}: \bar Q_{ijk}^{\tau}=\bar
Q_{ijk}=\bar Q_{jik}\}_{i,j,k \in \mathcal{S} }$ such that for all $i,\ j,\ k\in \mathcal{S}$, the following
inequalities  hold:
\begin{align}
&\quad\quad P_{k}=\bar{P}+D_{k}>0\label{e1h+},\ \ \ D-D_{k}\geq 0\allowdisplaybreaks[1]
\end{align}
\begin{align}
&(P_{k}A_{ij})^S\!+\!C_{i}^{\tau}(P_{k}+\beta D)C_{i}
\!+\!\bar Q_{ijk}<0\label{c61++}\allowdisplaybreaks[1]\\
&\quad\quad\quad\quad\quad\left[
\begin{array}{ccccc}
\bar Q_{11k} & \ldots &\bar Q_{1sk}\\
\vdots & \ddots & \vdots\\
\bar Q_{s1k} & \ldots &\bar Q_{ssk}
\end{array}
\right] > 0\label{c6+}
\end{align}
where $\beta$ is given in (\ref{bound2}).
\end{theorem}
Since its proof is similar to that of Theorem $1$, we omit it.
\begin{remark}
%Since there is the term $P_{k}A_{ij}$  in inequalities
%(\ref{c61++}), it is
%impossible to determine the control gain $K_j$ by some change of
%variables. For controllerdesign,
For determining the control gains $K_j$, pre-multipling and post-multipling $P_{k}^{-1}$ to
inequalities (\ref{c61++}) lead to for any $i,j,k \in\mathcal{S}$,
%{\setlength\abovedisplayskip{6pt}
%\setlength\belowdisplayskip{5pt}
\begin{eqnarray}
(A_{ij}P_{k}^{\!-\!1})^S\!+\!
P_{k}^{\!-\!1}C_{i}^{\tau}(P_{k}\!+\!\beta
D)C_{i}P_{k}^{\!-\!1}
\!+\!P_{k}^{\!-\!1}\bar Q_{ijk} P_{k}^{\!-\!1}\!<\!0.  \label{ae5}
\end{eqnarray}
One can see that the term $A_{ij}P_{k}^{-1}$
in (\ref{ae5}) makes it impossible for
obtaining the control gain $K_j$ via  change of variables.  To overcome the obstacle in
determining the control gains $K_j$, the
following lemma is adopted, based on which the term $AP$ can be decoupled by the additional matrix variable $\Omega$.
\end{remark}

\begin{lemma} \cite{GK2006} Let matrix $A \in \mathbb{R}^{n\times n}$ be given. If there exist
matrices $\Omega,W \in \mathbb{R}^{n\times n}$ and positive matrices
$P,R \in \mathbb{R}^{n\times n} $ such that the inequality
\begin{eqnarray*}
\left[
\begin{array}{ccc}
(A\Omega)^S+R+W & A\Omega & 0\\
 \ast & -\Omega^S & P-\Omega\\
 \ast & \ast &-R
\end{array}
\right] < 0\label{c21}
\end{eqnarray*}
holds, then we have $W+(AP)^S< 0$.
\end{lemma}
 Using Lemma $3$ to (\ref{ae5}),  a set of inequalities can be obtained, in which
 both $P_{k}$ and $P_{k}^{-1}$ appear simultaneously. Moreover,
 the structure constraint  (\ref{e1h+}) should be
satisfied for $P_{k}$. These make the determination of
 control gains $\{K_j\}_{j\in \mathcal{S}}$  difficult.
Motivated by cone complementarity linearization algorithm
in \cite{LElGhaoui1997}, the determination for the local gains $\left\{ K_i \right\}_{i\in\mathcal{S}}$
can be converted to a quadratic optimization problem with linear inequality constraints. So  Theorem $3$ can be
established.
\begin{theorem} \label{123a}
Under Assumption $1$, the equilibrium point of the It\^{o} closed-loop stochastic
system in (\ref{e6}) is stochastically asymptotically stable in
the large, if there exist matrices $\bar{P}$
and $D_{k}$ in the forms of (\ref{a3}), diagonal matrices $D$
and $\bar{D}$, symmetric matrices  $\bar P_{k}$, $R_{ijk}$, $Q_{ijk}$ with
$Q_{ijk}^{\tau}=Q_{ijk}=Q_{jik}$,  as well as  matrices
$M_{j}$, $\Omega_{j}$ being a set of solutions to the
quadratic optimization problem:
\begin{eqnarray}\mathrm{Min} \bigg \{ \sum^s_{k=1}  tr\left[ \bar
P_{k}(D_{k}+\bar{P})\right]+ tr(\bar{D}D) \bigg\}
\label{c230}
\end{eqnarray}
subject to the following constraints for all $i,j,k \in \mathcal{S}$:
\begin{align}
&  D\!-\!D_{j}\!\geq\! 0, \ \bar{P}\!+\!D_{k}\!>\!0, \ \left[
\begin{array}{cc}
\bar P_{k} & I \\
I & D_{k}+\bar{P}
\end{array}
\right] \!\geq \!0
\label{c23}\\
%\end{align}
%\begin{align}
&\Theta_2\!=\!\left[
\begin{array}{ccccc}
 Q_{11k} & \ldots & Q_{1sk}\\
\vdots & \ddots & \vdots\\
 Q_{s1k} & \ldots & Q_{ssk}
\end{array}
\right] \!>\! 0, \left[
\begin{array}{cc}
\bar{D} & I \\
I & D
\end{array}
\right] \!\geq\! 0\label{c23+}\\
%\end{align}
%\begin{align}
&\left[
\begin{array}{cccccc}
\Lambda_{ijk} & \Upsilon_{ij} & 0 & \bar P_{k}C_{i}^{\tau}& \bar P_{k}C_{i}^{\tau}\\
\ast & \!-\!\Omega_{j}^S & \bar P_{k}\!-\!\Omega_{j} &0&0\\
\ast & \ast & \!-\!R_{ijk}& 0 &0\\
\ast & \ast & \ast& \!-\!\bar P_{k}&0\\
\ast & \ast & \ast&\ast&\!-\!\beta^{\!-\!1}\bar{D}
\end{array}
\right] \!<\!0  \label{c23-}\\
&\Upsilon_{ij}\!=\! A_{i}\Omega_{j}\!+\!B_{i}M_{j}, \ \ \Lambda_{ijk}\!=\!\Upsilon_{ij}^S\!+\!R_{ijk}\!+\!Q_{ijk} \label{1-1}
\end{align}
where $\beta$ is given in (\ref{bound2}).
 In this case, the fuzzy local feedback gains are given by
\begin{equation}
K_{j}=M_{j}\Omega_{j}^{-1}.  \label{c22}
\end{equation}
\end{theorem}
\vspace{-0.3cm}
\begin{proof}
 Assume that there exist matrices $D_{k}$ and $\bar{P}$
in the form of (\ref{a3}), and matrices $M_{j}$, $\Omega_{j}$, $\bar
P_{k}$, $R_{ijk}$, $\{Q_{ijk}: Q_{ijk}^{\tau}=Q_{ijk}=Q_{jik}\}$
being a set of solutions to the minimization problem
 (\ref{c230})-(\ref{1-1}). It follows from the cone complementarity linearization algorithm
in \cite{LElGhaoui1997} that the
 matrices $D_{k}$,
 $\bar{P}$, $M_{j}$,
$\Omega_{j}$, $R_{ijk}$ and  $ Q_{ijk}$ satisfy
\begin{eqnarray}
\left.
\begin{array}{c}
\left [ \begin{array}{ccccccc}
\Lambda_{ijk}& \Upsilon_{ij} & 0 &  X_{k}C_{i}^{\tau}& X_{k}C_{i}^{\tau} \\
\ast & \!-\!\Omega_{j}^S & X_{k}\!-\!\Omega_{j} &0&0\\
\ast & \ast & \!-\!R_{ijk}& 0 &0\\
\ast & \ast & \ast& \!-\!X_{k}&0 \\
\ast & \ast & \ast&\ast&\!-\!D^{-1}\beta^{-1}
\end{array}
\right]\!<\!0 %\label{d11}\\
\\
\Theta_2\!>\! 0,\
\begin{array}{c}
\bar{P}\!+\!D_{k}\!>\!0, \ \  D\!-\!D_{j}\!\geq\! 0%\label{z11}
\end{array}
\end{array}
\right\}\label{cbv23}
\end{eqnarray}
where $X_{k}=P_{k}^{-1}$ with $P_{k}=\bar{P}+D_{k}$. Recalling
(\ref{e61}) and (\ref{c22}), the first expression of (\ref{cbv23})
becomes\\
\begin{equation}
\left [ \begin{array}{ccccc}
\tilde{\Lambda}_{ijk}& A_{ij}\Omega_{j} & 0 &  X_{k}C_{i}^{\tau}& X_{k}C_{i}^{\tau}  \\
\ast &\! -\!\Omega_{j}^S & X_{k}\!-\!\Omega_{j} &0&0\\
\ast & \ast & \!-\!R_{ijk}& 0 &0\\
\ast & \ast & \ast& \!-\!X_{k}&0\\
\ast & \ast & \ast& \ast &\!-\!D^{-1}\beta^{-1}
\end{array}
\right]\!<\!0  \label{acbv23}
\end{equation}
where $\tilde{\Lambda}_{ijk}=(A_{ij}\Omega_{j})^S+R_{ijk}+Q_{ijk}$, which together with the Schur complement equivalence yields
\begin{equation}
\left [ \begin{array}{ccc}
\Phi_{ijk}& A_{ij}\Omega_{j} & 0 \\
\ast & -\Omega_{j}-\Omega_{j}^{\tau} & X_{k}-\Omega_{j} \\
\ast & \ast & -R_{ijk}
\end{array}
\right] <0 \label{+acbv23}
\end{equation}
where
\begin{align*}
\Phi_{ijk}\!=\! (A_{ij}\Omega_{j})^S\!+\!R_{ijk}\! +\!Q_{ijk}\!+ \!X_{k}C_{i}^{\tau}(X_{k}^{\!-\!1}\!+\!\beta D)C_{i}X_{k}.
\end{align*}
Applying Lemma 3 to {(\ref{+acbv23})}, one has
\begin{align}
(A_{ij}X_{k})^S+Q_{ijk}+X_{k}C_{i}^{\tau}(X_{k}^{-1}+\beta
D)C_{i}X_{k}<0.\label{d12}
\end{align}
Pre-multiplying and post-multiplying $P_{k}$ to {(\ref{d12})} lead
to
\begin{equation}
(P_{k}A_{ij})^S\!+\!P_{k}Q_{ijk}P_{k}\!+\!
C_{i}^{\tau}(P_{k}+\beta
D)C_{i}\!<\!0.\label{e621}
\end{equation}
Letting
\begin{equation}
P_{k}Q_{ijk}P_{k}=\bar Q_{ijk} \label{ilvf++}
\end{equation}
and using $Q_{ijk}^{\tau}=Q_{ijk}=Q_{jik}$, one can easily check
that
\begin{equation*}
\bar Q_{ijk}^{\tau}=P_{k}Q_{ijk}^{\tau} P_{k} = P_{k}Q_{ijk}P_{k}
 =P_{k}Q_{jik}P_{k}=\bar Q_{jik}=
\bar Q_{ijk}.
\end{equation*}
Inequalities {(\ref{e621})} together with {(\ref{ilvf++})}  gives
\begin{equation}
(P_{k}A_{ij})^S+\bar Q_{ijk}+
C_{i}^{\tau}(P_{k}+\beta
D)C_{i}<0.\label{e62}
\end{equation}
It follows from (\ref{ilvf++}) and the first expression of (\ref{c23+})
that
\begin{align}
&\left[
\begin{array}{ccccc}
\bar Q_{11k} & \ldots &\bar Q_{1sk}\\
\vdots & \ddots & \vdots\\
\bar Q_{s1k} & \ldots &\bar Q_{ssk}
\end{array}
\right]\!=\!
\left[
\begin{array}{ccccc}
 P_{k}Q_{11k}P_{k} & \ldots & P_{k}Q_{1sk}P_{k}\\
\vdots & \ddots & \vdots\\
 P_{k}Q_{s1k}P_{k} & \ldots & P_{k}Q_{ssk}P_{k}
\end{array}
\right]\nonumber\\
%\end{align}
%\begin{align}
&=%\setlength{\arraycolsep}{0.2pt}
\left[
\begin{array}{ccccc}
 P_{k} &   &  & \\
 &  & \ddots & \\
 &  &  & P_{k}
\end{array}
\right]
\Theta_2
\left[
\begin{array}{ccccc}
 P_{k} &   &  & \\
 &  & \ddots & \\
 &  &  & P_{k}
\end{array}
\right]
> 0. \label{ylm+}
\end{align}
The conditions $P_{k}\!=\!\bar{P}\!+\!D_{k}>0$, $D\!-\!D_{j}\geq 0$ in
(\ref{c23}) together with (\ref{e62}) and (\ref{ylm+}) prove the
conclusion by  Theorem \ref{the+*-}.
\end{proof}
\begin{remark} One idea behind Theorem
\ref{123a}  is to decouple $A_{ij}X_{k}$ in (\ref{e621}) by introducing the
additional matrix variables $\Omega_{j}$ based on Lemma
3. Thus, by tolerating some conservativeness, the
nonlinear inequalities (\ref{d12}) can be expressed as
inequalities (\ref{acbv23}), where $B_{i}K_{j}\Omega_{j}$ in
$A_{ij}\Omega_{j}$ makes  the determination of the control gains $K_j$
feasible by change of variables. Another idea
behind this theorem   is that the nonlinear inequalities (\ref{cbv23})
in which both $X_{k}$ ( $X_{k}\!=\!P_{k}^{-1}$) and
$P_{k}$ appear, can be converted into
 the minimization problem  (\ref{c23})-(\ref{1-1}),
which can be solved by using the cone complementarity linearization
algorithm. However, from Lemma 3 one can see that
inequalities (\ref{acbv23}) are a sufficient condition for the
inequalities (\ref{d12}) holding, which means that the
existence of solutions for (\ref{d12}) does not guarantee the
existence of solutions for (\ref{acbv23}), so the
transformation from (\ref{d12}) to (\ref{acbv23})
leads to some conservativeness. In addition, it should be pointed
out that for fixed $j$ in (\ref{c23-})-(\ref{1-1}), the same
$\Omega_{j}$ is required to satisfy $s^{2}$ inequalities when $i,k$
range from $1$ to $s$, which inevitably results in some
conservativeness.
\end{remark}
%\begin{remark}
%Theorem $3$  shows that designing controller in (\ref{d2}) for the It\^{o} stochastic continuous T-S model in (\ref{a1}) can be reduced to the minimization problem of which the penalty function is (\ref{c230}) with linear constraints in (\ref{c23})-(\ref{1-1}).
%\end{remark}
Solving the quadratic optimization problem in Theorem $3$ is an iterative procedure  which is summarized in the following Iteration algorithm $1$.
It is clear when LMIs  (\ref{c23})-(\ref{1-1}) have no solution, so does the minimization problem  (\ref{c230})-(\ref{1-1}). For this reason, it is assumed that LMIs  (\ref{c23})-(\ref{1-1}) is feasible in this iteration algorithm.\\
Iteration algorithm $1$:
\begin{itemize}
\item [(1)] Given $\epsilon>0$ and a natural number $n_{max}$, determine a feasible solution
\[ I_0=\left(\bar{P}_{10},\cdots,\bar{P}_{s0},D_{10},\cdots,D_{s0},\bar{P}(0),\bar{D}_{0},D(0), \hat{I}_0
\right)\] to LMIs  (\ref{c23})-(\ref{1-1}) where $\hat{I}$  stands for other matrix variables in Theorem $3$ and the index 0's in these symbols denote these data appear in the initial iterative step. The iteration error is defined by
\[
E_0=\sum^s_{k=1}  tr\left[ \bar P_{k0}(D_{k0}+\bar{P}(0))\right]+ tr(\bar{D}_{0}D(0))-(s+1)n.
\]
If $\mid E_0 \mid < \epsilon$, then exit; or else  set $j=0$ and go to step $(2)$.
\item [(2)] For $j\geq 1$, assume that the data
\[ I_{j}=\left(\bar{P}_{1j},\cdots,\bar{P}_{sj},D_{1j},\cdots,D_{sj},\bar{P}(j),\bar{D}_{j},D(j), \hat{I}_j
\right)\]
are obtained in the $j$th iteration step. For data $I_j$ obtained in the $j$th  step, solve the LMI problem:
\begin{eqnarray*}&&\mathrm{Min} \left \{ \sum^s_{k=1}tr[\bar{P}_{kj}(D_{k}+\bar{P})+\bar{P}_{k}(D_{kj}+\bar{P}(j))]\right.\\
&&\left.+tr(\bar{D}_{j}D+\bar{D}D(j)) \right\}
\label{cc2}
\end{eqnarray*}
subject to (\ref{c23})-(\ref{1-1}) for determining its matrix variables $\{\bar{P}_i,\
   D_i, \bar{P}, \hat{I}\}_{i=1}^s$. Letting
\[D_{i,j+1}=D_i,\ \bar{P}_{i,j+1}=\bar{P}_i,\ \bar{P}(j+1)=\bar{P},\ \hat{I}_{j+1}=\hat{I}\]
  then we can obtain the data
 \begin{eqnarray*}
 I_{j\!+\!1}\!=\!(\bar{P}_{1,j\!+\!1}\cdots\bar{P}_{s,j\!+\!1},D_{1,j\!+\!1}\cdots
 D_{s,j\!+\!1},\bar{P}(j\!+\!1), \\
 \bar{D}_{j\!+\!1},D(j\!+\!1), \hat{I}_{j\!+\!1}).
\end{eqnarray*}
\item [(3)] The  $j$th iteration error is defined by
\begin{eqnarray*}
E_j&=&\sum^s_{k=1}  tr\left[ \bar P_{k,j+1}(D_{k,j+1}+\bar{P}(j+1))\right]\\&+& tr(\bar{D}_{j+1}D(j+1))-(s+1)n.
\end{eqnarray*}
If $\mid E_j \mid < \epsilon$ or $j= n_{max}$, then exit; or else  set $j=j+1$ and go to step $(2)$.
 \end{itemize}

\section{Numerical Examples}
\begin{example}
The It\^{o} stochastic T-S model based
control system has the following
fuzzy rules:
\par\setlength\parindent{3em}  \textbf{$R_{1}$}: $\ \ $If $x_{1}$ is $ F_{1}^{1}$ and $x_{2}$ is $
F_{2}^{1}$, then \newline \indent ~~~~~~~$dx=(A_{1}x+B_{1}u)dt+C_{1}x dW(t)$% \par\setlength\parindent{1.5em}
\newline \indent
  \textbf{$R_{2}$}:  $\ \ $If $x_{1}$ is $ F_{1}^{1}$ and $x_{2}$ is $
F_{2}^{2}$,  then \newline \indent ~~~~~~~ $dx=(A_{2}x+B_{2}u)dt+C_{2}x dW(t)$
\newline \indent \textbf{$R_{3}$}:  $\ \ $If $x_{1}$ is $ F_{1}^{2}$ and $x_{2}$ is $
F_{2}^{1}$, then \newline \indent ~~~~~~~  $dx=(A_{3}x+B_{3}u)dt+C_{3}x dW(t)$
\newline \indent \textbf{$R_{4}$}: $\ \ $If $x_{1}$ is $ F_{1}^{2}$ and $x_{2}$ is $
F_{2}^{2}$, then \newline \indent ~~~~~~~ $dx=(A_{4}x+B_{4}u)dt+C_{4}x dW(t)$
%\vspace{0.2cm}
\newline  where all the known matrices are listed as follows:
\end{example}
%\vspace{-1cm}
\begin{align*}
A_{1}&\!=\!\left[
\begin{array}{cc}
  \!-\!0.5  &       0\\
   \!-\!1   & 1.59
\end{array} \right],
A_{2}\!=\!\left[
\begin{array}{cc}
 \!-\!0.8  &  1\\
  0  &  1
\end{array} \right], B_{1}\!=\!\left[
\begin{array}{c}
\!-\!2\\
 3
\end{array} \right]\\
%\end{align*}
%\begin{align*}
A_{3}&\!=\!\left[
\begin{array}{cc}
  \!-\!0.9   & 1.09\\
  0   & 0.8
\end{array} \right],
A_{4}\!=\!\left[
\begin{array}{cc}
\!-\!1  &  1\\
0 &   0.4
\end{array} \right],
B_{2}\!=\!\left[
\begin{array}{c}
\!-\!2\\
 2
\end{array} \right],\\%\\
%\end{align*}
%\begin{align*}
 C_{1}&\!=\!\left[
\begin{array}{cc}
0.2 & 0\\
0 & \!-\!0.1
\end{array} \right],
C_{2}\!=\!\left[
\begin{array}{cc}
0.3 & 0\\
0 & \!-\!0.2
\end{array} \right], B_{3}\!=\!\left[
\begin{array}{c}
\!-\!1\\
 3
\end{array} \right]
\\
C_{3}&\!=\!\left[
\begin{array}{cc}
0.2 & 0\\
0 & \!-\!0.3
\end{array} \right], C_{4}\!=\!\left[
\begin{array}{cc}
 0.35   &      0\\
 0  & \!-\!0.1
\end{array} \right],
B_{4}\!=\!\left[
\begin{array}{c}
\!-\!1.5\\
 2.6
\end{array} \right].
\end{align*}
Then all ordinal numbers $\alpha_{ij},i=1, \cdots,4;j=1,2$ read
%\[\alpha_{11}=1,\quad \alpha_{12}=1,\quad\alpha_{21}=2,\quad \alpha_{22}=1,\quad \alpha_{31}=2,\quad
%\alpha_{32}=2,\quad \alpha_{41}=1,\quad \alpha_{42}=2. \]
{\setlength\abovedisplayskip{1pt}
\setlength\belowdisplayskip{5pt}
\begin{align*}
 &\alpha_{11}=1,\quad\alpha_{21}=1,\quad\alpha_{31}=2,\quad\alpha_{41}=2\\
 & \alpha_{12}=1,\quad\alpha_{22}=2,\quad\alpha_{32}=1,\quad\alpha_{42}=2
\end{align*}
with $s_1=s_2=2$. It is assumed that the membership functions
$\left\{w_j^i(x_j),\  i=1, 2;\ j=1,2 \right\}$ of the fuzzy sets
$\{F^i_j:  i,j=1,2\}$ in the fuzzy rules are given by
{\setlength\abovedisplayskip{6pt}
\setlength\belowdisplayskip{6pt}
\begin{align*}
   &w_1^1(x_1)=0.0169e^{- x_1^2},\quad w_2^1(x_2) =0.0024e^{-0.05 (x_2-3)^2}\\
  & w_1^2(x_1)=1-w_1^1(x_1),\quad w_2^2(x_2) =1-w_2^1(x_2).
\end{align*}
One has by  (\ref{normalization}) that
%{\setlength\abovedisplayskip{6pt}
%\setlength\belowdisplayskip{8pt}
\begin{align*}
&\mu_{1}^{\alpha_{11}}(x_{1}) =\mu_{1}^{\alpha_{21}}(x_{1}) =w_1^1(x_1)=0.0169e^{- x_1^2}\\
&\mu_{1}^{\alpha_{31}}(x_{1}) =\mu_{1}^{\alpha_{41}}(x_{1}) =w_1^2(x_1)=1-0.0169e^{-x_1^2}\\
&\mu_{2}^{\alpha_{12}}(x_{2}) =\mu_{2}^{\alpha_{32}}(x_{2})  =w_2^1(x_2)=0.0024e^{-0.05(x_2-3)^2}\\
&\mu_{2}^{\alpha_{22}}(x_{2}) =\mu_{2}^{\alpha_{42}}(x_{2})  =w_2^2(x_2)=1-0.0024e^{-0.05(x_2-3)^2}.
\end{align*}
It can be checked that Assumption 1 is satisfied  for all the
  $\mu_j^{\alpha_{ij}}(x_j),
 i=1, \cdots,4;\ j=1,2 $ with all  upper
bounds $ \beta_{ij},\  i=1, \cdots,4;\ j=1,2 $ listed in Table $1$.  %shown at the bottom of this page%原表格处
%\vspace{-0.8cm}
\begin{center}
\begin{tabular}{|c|c|c|c|c|}
\hline
 $\beta_{11}$ & $\beta_{21}$& $\beta_{31}$&$\beta_{41}$&$\beta_{12}$ \\
\hline
$0.0125$&$0.0125$ &$0.0125$ &$0.0125$ & $0.0125$\\
\hline
$\beta_{22}$& $\beta_{32}$ & $\beta_{42}$ & $\times$   & $\times$ \\
\hline
$0.0125$& $0.0125$ & $0.0125$& $\times$   & $\times$   \\
\hline
\end{tabular}
\vspace{0.2cm}
{ \\Table 1 : Upper bounds $\beta_{ij}$}
\end{center}
%\vspace{-0.5cm}
The objective is to design a  controller such that
the resulting It\^{o} closed-loop stochastic system  is
stochastically asymptotically stable. Solving the minimization
problem (\ref{c230})-(\ref{1-1}), we obtain a  set of
solutions as follows:
{\setlength\abovedisplayskip{5pt}
\begin{align*}
P_{1}&=%\setlength{\arraycolsep}{3pt}
\left[
\begin{array}{cc}
    0.8004  &  0.1785 \\
    0.1785  &  2.4479
\end{array} \right],\
 P_{2}=%\setlength{\arraycolsep}{3pt}
 \left[
\begin{array}{cc}
    0.8217  &  0.1785 \\
    0.1785  &  2.4479
\end{array} \right] \\
P_{3}&=%\setlength{\arraycolsep}{3pt}
\left[
\begin{array}{cc}
  0.8217  &  0.1785\\
   0.1785  &  2.6144
\end{array} \right],
 P_{4}=%\setlength{\arraycolsep}{3pt}
\left[
\begin{array}{cc}
    0.8004  &  0.1785 \\
   0.1785  &  2.6144
\end{array}\right]\\
\bar P_{1}&\!=\!%\setlength{\arraycolsep}{3pt}
\left[
\begin{array}{cc}
    1.2701 &  \!-\!0.0926 \\
  \!-\!0.0926  &  0.4153
\end{array} \right],
\bar P_{2}\!=\!%\setlength{\arraycolsep}{3pt}
\left[
\begin{array}{cc}
    1.2365  & \!-\!0.0902 \\
    \!-\!0.0902  &  0.4151
\end{array} \right]\\
\bar P_{3}&\!=\!%\setlength{\arraycolsep}{3pt}
\left[
\begin{array}{cc}
   1.2353  & \!-\!0.0843 \\
   \!-\!0.0843  &  0.3883
\end{array} \right],
\bar P_{4}\!=\!%\setlength{\arraycolsep}{3pt}
\left[
\begin{array}{cc}
   1.2687  & \!-\!0.0866\\
   \!-\!0.0866 &   0.3884
\end{array}\right]\\%\allowdisplaybreaks[1]\\
 D&=%\setlength{\arraycolsep}{3pt}
 \left[
\begin{array}{cc}
     1.5380   &      0\\
      0   &  2.6231
\end{array} \right],
 \bar{D}=%\setlength{\arraycolsep}{3pt}
\left[
\begin{array}{cc}
   0.6502    &     0 \\
        0   & 0.3812
\end{array} \right]\\
\Omega_{1}&\!=\!%\setlength{\arraycolsep}{3pt}
\left[
\begin{array}{cc}
    1.7249  &  0.0649\\
   \!-\!0.2823  &  0.3543
\end{array} \right],\
 \Omega_{2}\!=\!%\setlength{\arraycolsep}{3pt}
 \left[
\begin{array}{cc}
    1.5810  &  0.1353 \\
   \!-\!0.3666  &  0.2958
\end{array} \right]\\
%\end{align*}\begin{align*}
 \Omega_{3}&\!=\!%\setlength{\arraycolsep}{3pt}
 \left[
\begin{array}{cc}
   1.7120 &   0.0859 \\
   \!-\!0.3078  &  0.3397
\end{array} \right],\
\Omega_{4}\!=\!%\setlength{\arraycolsep}{3pt}
\left[
\begin{array}{cc}
    1.9383 &   0.0697 \\
   \!-\!0.3188  &  0.3510
\end{array}\right]\\
 M_{1}&\!=\!%\setlength{\arraycolsep}{3pt}
 \left[
\begin{array}{rr}
  0.4650 &  \!-\!0.3815
\end{array} \right],\
  M_{2} \!=\!%\setlength{\arraycolsep}{3pt}
  \left[
\begin{array}{rr}
    0.3715 &  \!-\!0.3311
\end{array} \right]\\%\allowdisplaybreaks[1]\\
 M_{3}&\!=\!%\setlength{\arraycolsep}{3pt}
  \left[
\begin{array}{rr}
       0.3924  &  \!-\!0.3659
\end{array} \right],\
  M_{4}\!=\!%\setlength{\arraycolsep}{3pt}
  \left[
\begin{array}{rr}
    0.4414 &  \!-\!0.3802
\end{array}\right].
\end{align*}
By Theorem \ref{123a}, the local feedback gains are given by
\begin{align*}
K_{1}&\!=\!%\setlength{\arraycolsep}{3pt}
\left[
\begin{array}{rr}
 0.0906  & \!-\!1.0936
\end{array}
\right],\
K_{2}\!=\!%\setlength{\arraycolsep}{3pt}
\left[
\begin{array}{rr}
 \!-\!0.0222  & \! -\!1.1092
\end{array}
\right]\\
%\end{align*}\begin{align*}
K_{3}&\!=\!%\setlength{\arraycolsep}{3pt}
\left[
\begin{array}{rr}
0.0340  & \!-\!1.0855
\end{array}
\right],\
K_{4}\!=\!%\setlength{\arraycolsep}{3pt}
\left[
\begin{array}{rr}
0.0480  & \!-\!1.0927
\end{array}
\right].
\end{align*}
To illustrate the results in Theorem
\ref{123a}, Monte Carlo simulations have been carried out by using a
discretization approach \cite{higham2001,YugangNiu2005,zhouren2010}. Some initial parameters are given as follows:
simulation interval $t\in[0,T]$ with $T=15$, initial values
$x_a(0)\!=\!(\!-\!7,3)^\tau, x_b(0)\!=\!(\!-\!5,5)^\tau, x_c(0)\!=\!(\!-\!5,10)^\tau, x_d(0)\!=\!(12,10)^\tau$, normally distributed variance
$\delta t=T/N$ with $N=2^{8}$, step size $\Delta t=R \delta t$ with
$R=2$. Figures (a)-(d) depict the state responses along 2, 10, 30 and 50 individual
Wiener process paths respectively, as well as their mean values over
these paths. These figures also show the stochastic stability of the  It\^{o} closed-loop stochastic
system in (\ref{e6}).
\begin{figure}[H]
%\hspace{0.8cm}
 \subfigure[ \indent~~State responses of $x_{1}$ and $x_{2}$ along 2 paths  of  Wiener   process ]{\scalebox{1}[0.8]
{\includegraphics[width=9.0cm,height=4.2cm, trim=0 10 0
0]{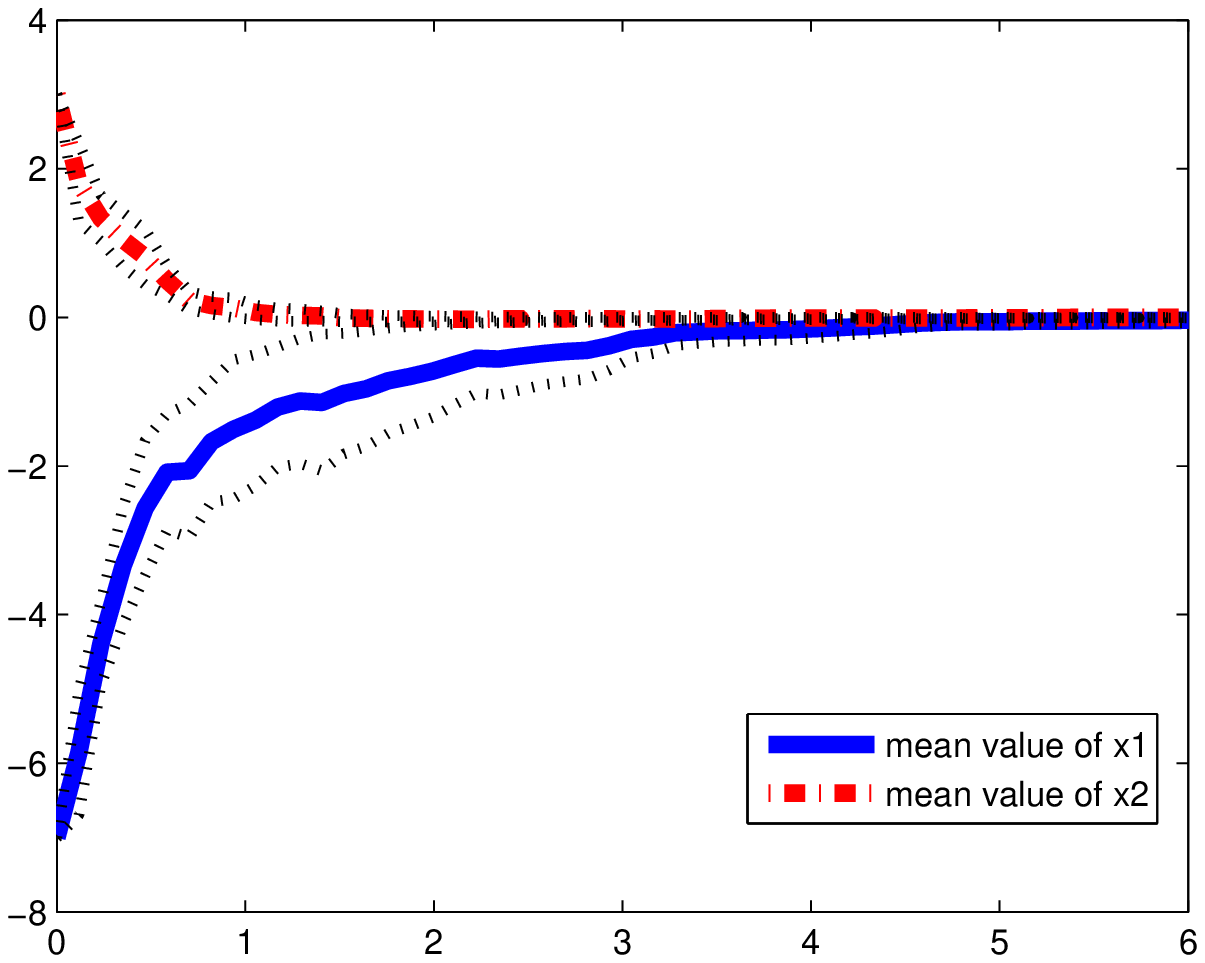}}}
\subfigure[ \indent ~~State responses of $x_{1}$ and $x_{2}$ along 10 paths  of  Wiener \  process ]{\scalebox{1}[0.8]
{\includegraphics[width=9.0cm,height=4.2cm, trim=0 10 0
0]{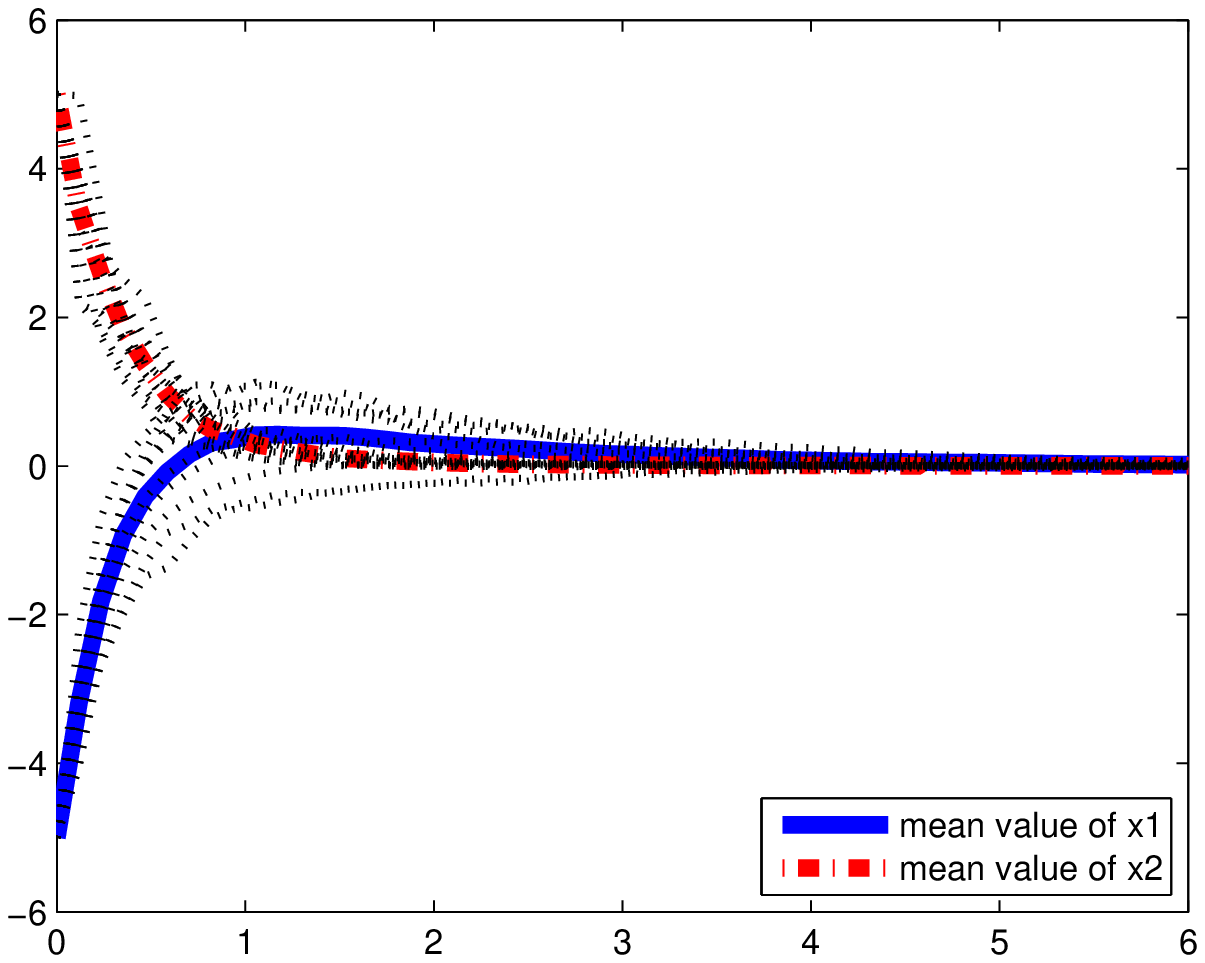}}}
\end{figure}
\begin{figure}[H]
\subfigure[ \indent~~ State responses of $x_{1}$ and $x_{2}$  along 30 paths  of  Wiener\  process ]{\scalebox{1}[0.8]
{\includegraphics[width=9.0cm,height=4.2cm, trim=0 10 0
0]{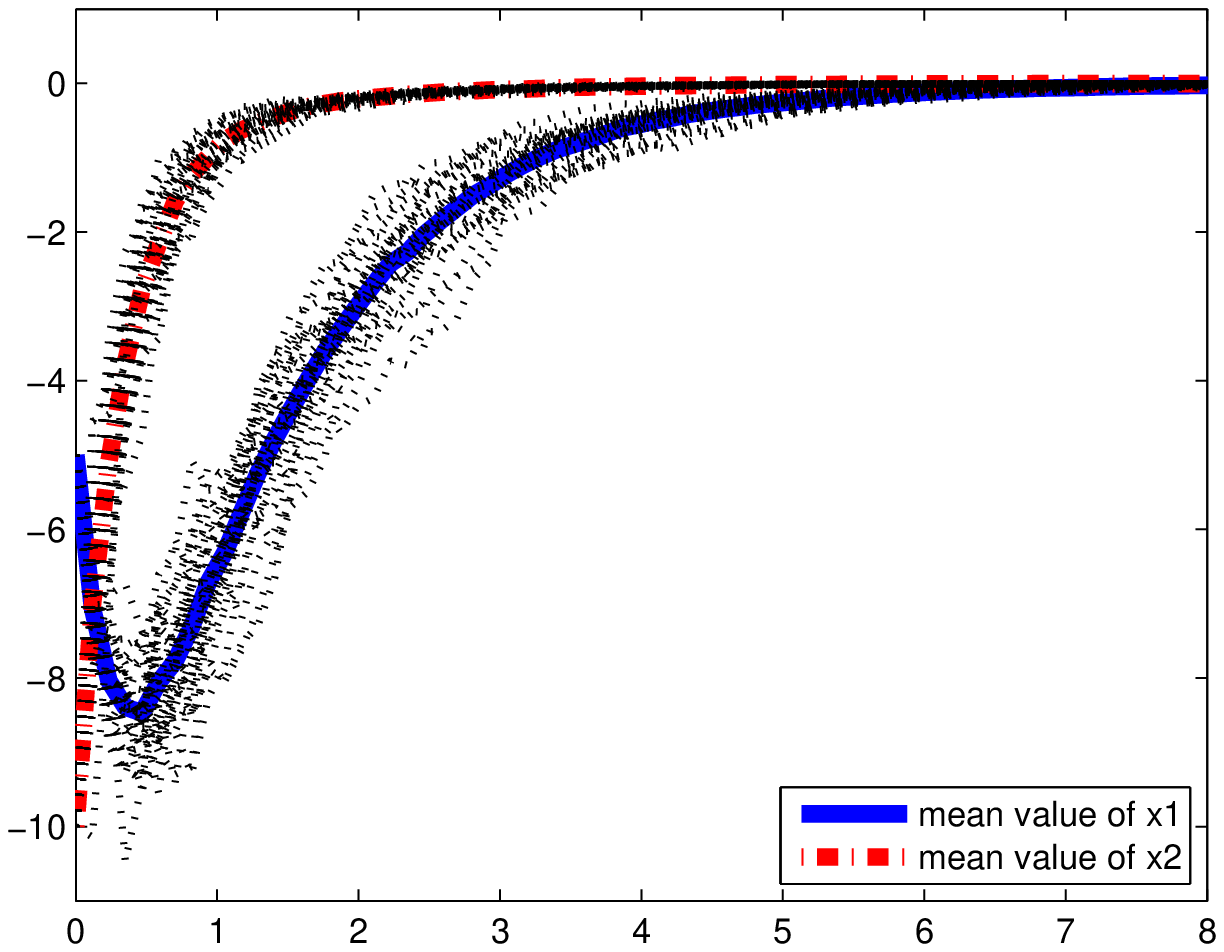}}}
 \subfigure[ \indent~~State responses of $x_{1}$ and $x_{2}$ along 50 paths  of  Wiener \  process ]{\scalebox{1}[0.8]
{\includegraphics[width=9.0cm,height=4.2cm, trim=0 10 0
0]{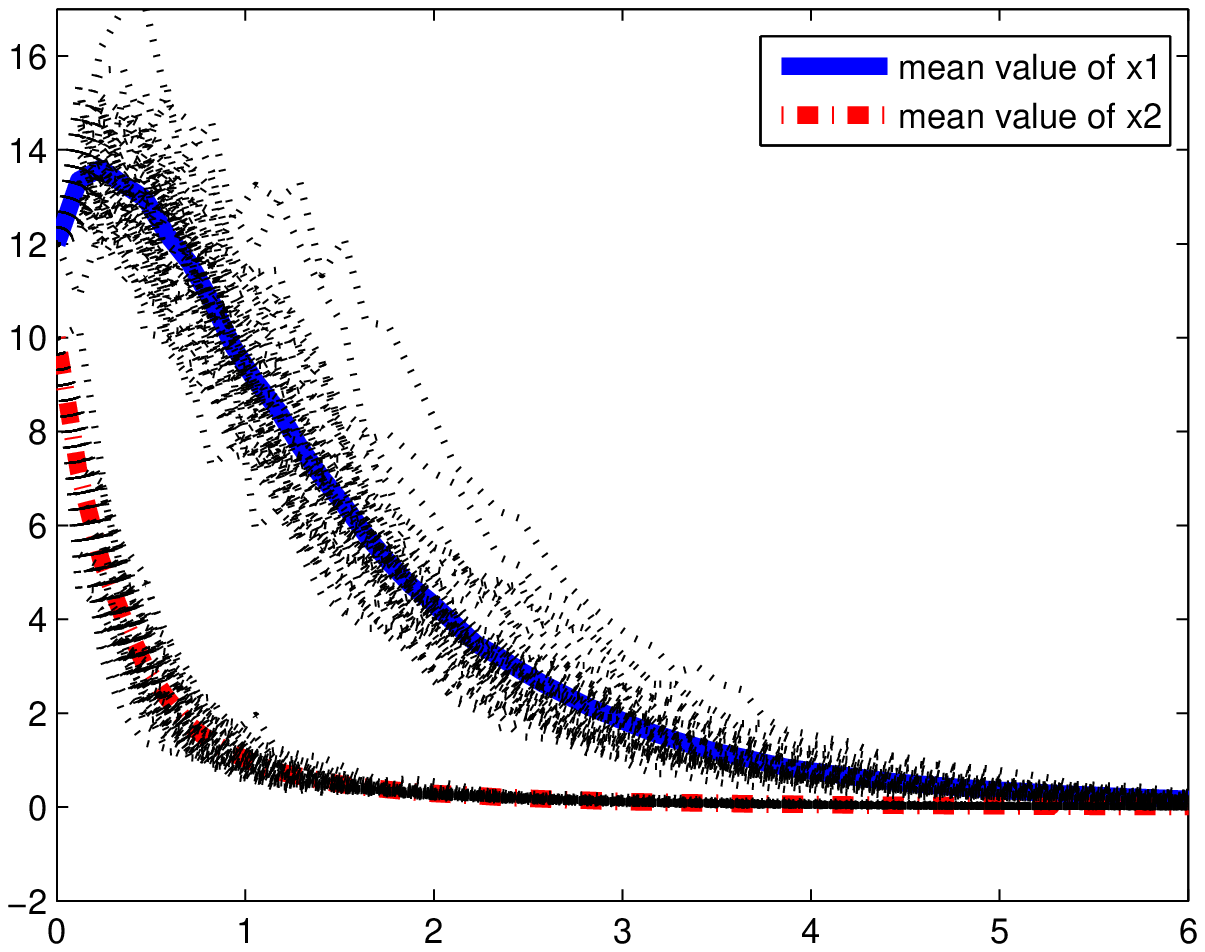}}}
\end{figure}

\section{Conclusion}
\noindent\ \ \ \ \   In this section, the main results and  the major features of this paper are summarized as follows.
Line integral approach to stability analysis and stabilization is proposed for It\^o  stochastic  T-S models. The stochastic stability analysis of  these models  needs to handle  Hessian matrix of the line integral Lyapunov  function. One can see from (\ref{aa4}) that the Hessian matrix can be decomposed into two parts: one part is the weighting sum of a set of positive definite matrices $P_i$; Another is a sum of $x$ dependent matrices $\frac{\partial h_{i}}{\partial x}x^{\tau}D_i$ with rank one. Invoking the property in (\ref{pdfra}) of matrix with rank one and applying Lemma $1$ lead to the upper bound $\beta y^\tau  D y$ of $y^\tau\sum^s_{i=1}\frac{\partial h_{i}(x)}{\partial x}x^{\tau}D_iy$. Then an estimate of  $\mathcal {L}V$  can also be achieved.
It has been shown that the line integral based analysis results of stability  are more general than those ones derived from quadratic Lyapunov function due to the fact that a quadratic Lyapunov function is a special line integral.
By facilitating the cone complementarity linearization algorithm  and introducing some additional matrix variables, a  line integral approach is  proposed for stabilization of the It\^{o} stochastic T-S model (\ref{a1}).

\noindent\ \ \ \ \ There is also some conservativeness
about the  proposed approach to stabilization, one source of which  may be that to make controller design feasible, we have to
impose certain restrictions on the additional matrix variables. Furthermore, using Lemma 3 to derive the stabilization result in  Theorem \ref{123a} also results in some conservativeness.

\noindent\ \ \ \ \ By using the techniques in this paper, the results on the  stability and
stabilization of the It\^{o} stochastic T-S model  can
be readily extended to the fuzzy systems with uncertainties of the
norm-bounded or linear fractional types. In addition, the problem of
output feedback control for this class of systems with (or without)
time-delay deserves further attention.

\end{document}